\documentclass[aps,amssymb,amsmath,amsfonts,superscriptaddress]{revtex4}

\usepackage{graphicx}
\usepackage{bm,bbm}
\usepackage{epstopdf}

\usepackage{amsthm}
\usepackage{amsmath}
\usepackage{color}

\usepackage{hyperref}

\newcommand{\1}{\mathbbm{1}}

\usepackage[T1]{fontenc}
\usepackage[english]{babel}
\usepackage[latin2]{inputenc}

\newtheorem{proposition}{Proposition}
\newtheorem{lemma}{Lemma}
\newtheorem{remark}{Remark}
\newtheorem{theorem}{Theorem}

\newtheorem{corollary}[proposition]{Corollary}

\usepackage{amsopn}
\usepackage{amssymb}
\usepackage{hyperref}
\usepackage{listings}
\usepackage{subfig}
\usepackage{enumerate}
\usepackage{todonotes}
\usepackage{ dsfont }

\newcommand{\etal}{\emph{et al.}}

\newtheorem{definition}{Definition}

\DeclareMathOperator{\tr}{Tr}
\DeclareMathOperator{\Tr}{Tr}

\DeclareMathOperator{\diag}{diag}

\newcommand{\R}{\ensuremath{\mathbb{R}}}

\newcommand{\E}{\ensuremath{\mathbb{E}}}

\newcommand{\C}{\ensuremath{\mathbb{C}}}

\newcommand{\ket}[1]{\ensuremath{|#1\rangle}}
\newcommand{\bra}[1]{\ensuremath{\langle#1|}}
\newcommand{\ketbra}[2]{\ensuremath{\ket{#1}\bra{#2}}}
\newcommand{\proj}[1]{\ensuremath{\ketbra{#1}{#1}}}
\newcommand{\braket}[2]{\ensuremath{\langle{#1}|{#2}\rangle}}

\newcommand{\CC}{\mathcal{C}}

\begin{document}
\title{Generating random quantum channels}

\author{Ryszard Kukulski}
\email{rkukulski@iitis.pl}
\affiliation{Institute of Theoretical and Applied Informatics, Polish Academy
of Sciences, Ba{\l}tycka 5, 44-100 Gliwice, Poland}
\author{Ion Nechita}
\affiliation{CNRS, Laboratoire de Physique Th\'{e}orique, IRSAMC,
Universit\'{e} de Toulouse, UPS, F-31062 Toulouse, France}
\author{{\L}ukasz Pawela}
\affiliation{Institute of Theoretical and Applied Informatics, Polish Academy
of Sciences, Ba{\l}tycka 5, 44-100 Gliwice, Poland}
\author{Zbigniew Pucha{\l}a}
\affiliation{Institute of Theoretical and Applied Informatics, Polish Academy
of Sciences, Ba{\l}tycka 5, 44-100 Gliwice, Poland}
\affiliation{Faculty of Physics, Astronomy and Applied Computer Science,
Jagiellonian University, ul. {\L}ojasiewicza 11,  30-348 Krak{\'o}w, Poland}
\author{Karol {\.Z}yczkowski}
\affiliation{Faculty of Physics, Astronomy and Applied Computer Science,
Jagiellonian University, ul. {\L}ojasiewicza 11,  30-348 Krak{\'o}w, Poland}
\affiliation{Center for Theoretical Physics, Polish Academy of Sciences,\\ al.\
Lotnik\'ow 32/46, 02-668 Warszawa, Poland}
\date{May 19, 2021}

\begin{abstract}
Several techniques of generating random quantum channels, which act on the set
of $d$-dimensional quantum states, are investigated. We  present three
approaches to the problem of sampling of quantum channels and show they are 
mathematically equivalent. We discuss under which conditions they give the 
uniform, Lebesgue measure on the convex set of quantum operations and compare 
their advantages and computational complexity, and demonstrate  which of them is
particularly suitable for numerical investigations. Additional results focus on
the spectral gap and other spectral properties of random quantum channels and
their invariant  states. We compute mean values of several quantities
characterizing a given quantum channel, including its unitarity, the average
output purity and the $2$-norm coherence of a channel, averaged over  the entire
set of the quantum channels with respect to the uniform measure. An ensemble of
classical stochastic matrices obtained due to super-decoherence of random
quantum stochastic maps is analyzed and their spectral properties are studied
using the Bloch representation of a classical probability vector.
\end{abstract}

\maketitle 
%%%%%%%%%%%%%%%%%%%%%%%%%%%%%%%
\section{Introduction}
%%%%%%%%%%%%%%%%%%%%%%%%%%%%%%%%%%%%%%%%%%%%%
Quantum dynamics corresponding to classically chaotic systems can be described
by suitable ensembles of random matrices. In the case of autonomous quantum
systems one mimics Hamiltonians with the help of ensembles of random Hermitian
matrices invariant with respect to certain transformations. Depending on the
symmetry properties of the system investigated one uses orthogonal, unitary or
symplectic ensembles \cite{Ha10}. In the case of time-dependent, periodically
driven systems, corresponding unitary evolution operators can be described by
one of three circular ensembles of Dyson \cite{Me04,Fo10}.

The problem becomes more complex in a physically important situation of
\emph{open quantum systems} interacting with an environment \cite{AL00,BP07} and
dissipative quantum systems \cite{We08,Br01,WM10}. Time evolution of the
principal system coupled with an ancillary system is non-unitary, and the
spectrum of the associated evolution operators is contained in the unit disk.
Dynamics of dissipative chaotic quantum systems can be described by the master
equation \cite{Ha10,Br01}, while properties of chaotic scattering \cite{LW91}
can be explained with help of suitable non-unitary ensembles of random matrices
\cite{HILSS92,Sch16}.

A very general  scheme of a continuous time evolution of an open quantum system
can be described by the celebrated Gorini-Lindblad-Kossakowski-Sudarshan
equation \cite{GKS76,Li76,CP17}. In the case of discrete dynamics one uses the
notion of quantum operation or \emph{quantum channel}: a completely positive,
trace preserving linear map which sends the set of mixed quantum states into
itself \cite{BZ17}. The set of quantum operations acting on density matrices of
a given size $d$ is convex and compact, and it is easy to show that the flat
Lebesgue measure in this set is generated by the Hilbert-Schmidt distance in the
set of dynamical matrices which determine the maps.

Such an ensemble of random quantum operations was introduced in \cite{BCSZ09}.
The goal of this work is to introduce and study three families of ensembles of
quantum operations, motivated by different structural characterizations of
quantum channels. We show that these probability measures agree for specific
values of the parameters, and we provide practical algorithms to generate such
maps numerically. We argue that the most computationally efficient procedure to
generate a random quantum channel is to sample its Kraus operators, see
procedure \textbf{b)} from Section \ref{sec:different-measures}. Furthermore, we
also present several results on properties of random operations: we analyze
spectral properties of a generic superoperator and the corresponding invariant
state and show that in the limit of a large dimension $d$ a typical channel is
close to the unital completely depolarizing channel. As a closely related
subject we discuss properties of random stochastic and bistochastic matrices
which can be considered as classical analogues of random quantum stochastic
maps. The present work can be considered as a complementary to the earlier
contribution \cite{ZPNC11}, as it extends the study of random quantum states for
random operations.
Note also that the techniques of random matrices are frequently used in the
theory of quantum information \cite{CN16,AS17}. In particular, the seminal
result of Hastings \cite{Ha09} establishing non-additivity of minimum output
entropy of quantum channels was obtained with use of a certain class random
quantum operations.

\bigskip

This work is organized as follows. In Section~\ref{sec:2} we set the scene and
introduce concepts and notation necessary for presenting our results.
Section~\ref{sec:different-measures} describes the several equivalent sampling
methods for random quantum operations, as well as their relation to the flat
measure, while in subsequent Section \ref{sec_distr} the distribution of output
states of a random channel is investigated. Next, in
Section~\ref{sec:spectrum-random-quantum-channels} we discuss the spectral
properties of the superoperator representing a quantum channel. Furthermore, we
demonstrate that the Bloch representation of a quantum operation is equivalent
with the Fano form of the corresponding bi-partite Jamio{\l}kowski state, so
certain properties of a random superoperator are related to properties of a
correlation matrix, which describes correlations between two subsystems a
bi-partite random state. Effects of super-decoherence and a relation between
random quantum stochastic maps and the classical random stochastic matrices are
discussed in  section \ref{sec:classical}, while in Section  \ref{sec:inv}
spectral properties of invariant states of random quantum channels are analyzed.
The Appendices from A to E contain two technical lemmas and proofs of
propositions formulated in the main body of the work. The Appendix F provides a
short review of several other ensembles of random stochastic and bistochastic
quantum maps.

%%%%%%%%%%%%%%%%%%%%%%%%%%%%%%%%
\section{Setting the scene}\label{sec:2}
%%%%%%%%%%%%%%%%%%%%%%%%%%%%%%%%

Quantum channels, the most general transformations of quantum states allowed by
the axioms of quantum mechanics, are modeled by linear maps between matrix
algebras satisfying certain positivity and trace preservation properties. There
are several equivalent characterizations of quantum channels, each having its
own merit depending on the point of view we want to take. We summarize them in
the following theorem (see, e.g.~\cite[Corollary 2.27]{Wat18}).

\begin{theorem}\label{thm:def-quantum-channel} A linear map $\Phi: M_{d_1}(\C)
\to M_{d_2}(\C)$ is called a \emph{quantum channel} (or quantum operation) if
any of the following equivalent conditions is satisfied
\begin{enumerate}
	\item The map $\Phi$ is both
	\begin{itemize}
		\item \emph{completely positive}: for all $n \geq 1$, $\Phi \otimes
		\operatorname{id}_n : M_{d_1 n}(\C) \to M_{d_2 n}(\C)$ is positive
		(i.e.~maps positive semidefinite matrices to positive semidefinite
		matrices)
	\end{itemize}
	and
	\begin{itemize}
		\item \emph{trace preserving}: for all $X \in M_{d_1}(\C)$, $\Tr \Phi(X)
		= \Tr X$.
	\end{itemize}
	\item The map $\Phi$ admits a \emph{Kraus decomposition}:
	\begin{equation}\label{eq:def-Kraus-decomposition}
	\forall X \in M_{d_1}(\C), \quad \Phi(X) = \sum_{i=1}^r A_i X A_i^\dagger,
	\end{equation}
	for matrices $A_1, \ldots, A_r \in M_{d_2 \times d_1}(\C)$, called
	\emph{Kraus operators}. The matrices $A$ satisfy the identity resolution,
	$\sum_{i=1}^r A_i^{\dagger} A_i= \mathbbm{1}_{d_1}$ corresponding to the
	fact that $\Phi$ is trace preserving.
	\item The map $\Phi$ admits a \emph{Stinespring dilation}: there exists, for some positive integer $n$, an isometry $V : \C^{d_1} \to \C^{d_2} \otimes \C^n$ such that
	\begin{equation}\label{eq:def-Stinespring-dilation}
	\forall X \in M_{d_1}(\C), \quad \Phi(X) = [\operatorname{id}_{d_2} \otimes \Tr_{n}](VXV^\dagger).
	\end{equation}
	\item The \emph{Choi matrix} of $\Phi$
	\begin{equation}\label{eq:def-Choi-matrix}
	J_\Phi := \sum_{i,j=1}^{d_1} \Phi(\ketbra{i}{j}) \otimes \ketbra{i}{j} \in M_{d_2}(\C) \otimes M_{d_1}(\C)
	\end{equation}
	is positive semidefinite and has partial trace
  \begin{equation}
    [\Tr_{d_2} \otimes \operatorname{id}_{d_1}](J_\Phi) = \mathbbm{1}_{d_1}.
  \end{equation}
\end{enumerate}
\end{theorem}

A few remarks are in order regarding the result above. First, it suffices to
check the complete positivity condition for $n = d_1$. Second, regarding the
Kraus decomposition, the smallest positive integer $r$ for which $\Phi$ admits a
Kraus decomposition \eqref{eq:def-Kraus-decomposition} is called the \emph{Choi
rank} of $\Phi$ and is denoted by $\operatorname{rk}_C(\Phi)$. The Choi rank is
a measure of the noisiness of the channel $\Phi$, varying from $r=1$ for a
\emph{unitary conjugation} ($d_1=d_2=d$):
\begin{equation}
  \Phi_U(X) = UXU^*, \quad \text{ where $U \in \mathcal U(d)$ is a unitary matrix}
\end{equation}
to $r=d_1d_2$ for the \emph{completely depolarizing channel}
\begin{equation}
  \Phi_*(X) = \Tr X \frac{\mathbbm{1}_{d_2}}{d_2}.
\end{equation}

In the Stinespring dilation formulation, the Hilbert space $\C^n$ is commonly
termed the \emph{environment}. In the special case when $d_1$ divides $d_2n$,
equation \eqref{eq:def-Stinespring-dilation} can be rewritten as
\begin{equation}
  \Phi(X) =[\operatorname{id}_{d_2} \otimes \Tr_{n}]\left(U(X \otimes \ketbra{0}{0})U^\dagger\right)
\end{equation}
with $U \in \mathcal U(d_2n)$ a unitary operator and $\ket{0}$ a
$d_2n/d_1$-dimensional unit vector. The isometry $V$ defining the quantum
channel $\Phi$ is then a truncation of the unitary operator $U$. The environment
size $n$ can be taken to be $n=d_1d_2$ without loss of generality, and the
minimal $n$ for which a decomposition \eqref{eq:def-Stinespring-dilation} exists
is equal to the Choi rank $\operatorname{rk}_C(\Phi)$ of the channel.

The final characterization in Theorem \ref{thm:def-quantum-channel} is an
instance of the \emph{Jamio{\l}kowski-Choi isomorphism}: with any linear map
$\Phi: M_{d_1}(\C) \to M_{d_2}(\C)$ one can associate a matrix $J_\Phi \in
M_{d_2}(\C) \otimes M_{d_1}(\C)$; this isomorphism has the following properties:
\begin{itemize}
	\item maps preserving self-adjointness are mapped to self-adjoint matrices,
	\item completely positive maps are mapped to positive semidefinite matrices,
	\item trace preserving maps are mapped to matrices $J_\Phi$ satisfying
	$[\Tr_{d_2} \otimes \operatorname{id}_{d_1}](J_\Phi) = \mathbbm{1}_{d_1}$,
	\item unital maps (i.e.~$\Phi(\mathbbm{1}_{d_1}) = \mathbbm{1}_{d_2}$) are
	mapped to matrices $J_\Phi$ satisfying $[\operatorname{id}_{d_2} \otimes
	\Tr_{d_1}](J_\Phi) = \mathbbm{1}_{d_2}$.
\end{itemize}
The Choi matrix $J_\Phi$ can be written with the help of the \emph{maximally
entangled state}, $  \ket \Omega = \frac{1}{\sqrt{d_1}} \sum_{i=1}^{d_1}
\ket{ii} \in \C^{d_1} \otimes \C^{d_1}  $ as
\begin{equation}
  J_\Phi = d_1[\Phi \otimes \operatorname{id}_{d_1}](\ketbra \Omega \Omega).
  \label{jamiol}
\end{equation}
The rescaled Choi matrix, $J_{\Phi}/d_1$, with unit trace is also called {\sl
Jamio{\l}kowski state}, which explains the notation used.

Moreover, from the Kraus decomposition $\Phi(X) = \sum_{i=1}^r A_i X
A_i^\dagger$ of the channel $\Phi$ one can calculate the Choi matrix using the
vectorization notation,
\begin{equation}
J_\Phi=\sum_{i=1}^r |A_i\rangle\rangle \langle \langle A_i|,
\label{choi2}
\end{equation}
where $|A\rangle\rangle = \sum_{i=1}^{d_1} A \ket{i} \otimes \ket{i}$ denotes
the vector of length $d_2d_1$ obtained by reshaping a given matrix $A$ of size
$d_2 \times d_1$. Let us discuss now the matrix of $\Phi$, viewed as a linear
map between the vector spaces $M_{d_1}(\C) \cong \C^{d_1^2}$ and $M_{d_2}(\C)
\cong \C^{d_2^2}$. We also denote by $\Phi$ this matrix, usually called the
\emph{superoperator}; we have, in terms of the Kraus operators,
\begin{equation}
\Phi=\sum_{i=1}^r A_i\otimes \bar{A}_i,
\label{Phisum}
\end{equation}
where bar denotes the complex conjugation. In terms of the Choi matrix, the
superoperator reads $\Phi = J_\Phi^{\mathrm R}$, where $\mathrm R$ denotes the
transformation of \emph{matrix reshuffling} (or realignment) -- see
\cite[Chapter 10.2]{BZ17}:
\begin{equation}\label{eq:def-reshuffling}
\left( \ketbra i j \otimes \ketbra k l \right)^{\mathrm R} = \ketbra i k \otimes \ketbra j l.
\end{equation}

The linear map $\Phi: M_{d_1}(\C) \to M_{d_2}(\C)$ admits an adjoint
$\Phi^\dagger : M_{d_2}(\C) \to M_{d_1}(\C)$ for the Hilbert-Schmidt scalar
product on complex matrices $\braket{A}{B} = \Tr(A^\dagger B)$. If $\Phi$ is a
quantum channel, the map $\Phi^\dagger$ is still completely positive, but the
trace preservation property of $\Phi$ is converted to unitality:
$\Phi^\dagger(\mathbbm{1}_{d_2})  = \mathbbm{1}_{d_1}$. Quantum channels which
are unital are called \emph{bistochastic}, since they satisfy both normalization
conditions; the name is a reference to the classical situation where row- and
column-stochastic matrices are called bistochastic or doubly stochastic.

An  interested reader can find modern expositions of these results and many
developments in monographs such as \cite[Chapter 8]{NC10}, \cite[Chapters 10 and
11]{BZ17}, or \cite[Chapter 2.2]{Wat18}.

\bigskip

Some basic facts from the theory of random matrices will be relevant in this
paper. For an in-depth introduction, we refer the reader to the classical
textbook \cite{Me04} or to modern presentations \cite{AGZ10, MSp17}. In this
work we are going to use the following ensembles of random matrices:

\begin{itemize}
	\item the {\bf real Ginibre ensemble} consisting of matrices with
	independent and identically distributed real standard Gaussian entries
	$G_\R$.
	\item the {\bf complex Ginibre matrices} consisting of matrices $G$  with
	independent complex entries distributed according to the standard complex
	Gaussian distribution \cite{Gi65}. One can write $G = (G_\R' + i
	G_\R'')/\sqrt 2$, where $G_\R'$ and $G_\R''$ are independent real Ginibre
	matrices. Note that the Ginibre matrices can be rectangular, and they are
	normalized as $\mathbb E[\Tr(GG^{\dagger})]=d_1d_2$. In the square case
	($d_1=d_2=d$), the spectrum of a normalized complex Ginibre matrix $G /
	\sqrt d$ covers uniformly the unit disk, a result called the \emph{circular
	law of Girko} \cite{Gi65, Gi84,Fo10}.
	\item {\bf Gaussian Unitary ensemble} (GUE) of Hermitian (self-adjoint)
	matrices, invariant with respect to the unitary group, contains matrices
	$H=(G+ G^{\dagger})/\sqrt{2}$. In the limit of large matrix dimension $d$,
	the spectrum of normalized GUE matrices $H/\sqrt d$ converges to Wigner's
	semicircle distribution \cite[Chapter 2]{AGZ10}.
	\item {\bf the complex Wishart ensemble} of parameters $(d,s)$ consisting
	of matrices $W = GG^\dagger$, where $G$ is a rectangular $d \times s$
	random matrix from the complex Ginibre ensemble. By construction, Wishart
	matrices are positive semidefinite. The parameter $s$ can be chosen to be
	any real number in the set $\{1, 2, \ldots, d-1\} \cup [d, \infty)$. In the
	scaling limit where $d,s \to \infty$ in such a way that $s/d$ converges to
	a constant value $c >0$, the spectrum of $W/d$ converges to the
	\emph{Mar\v{c}enko-Pastur distribution} of parameter $c$ \cite[Definition
	11]{MSp17}.
	\item {\bf Circular unitary ensemble} (CUE) consisting of random unitary
	matrices distributed according to the Haar measure on the unitary group
	\cite{Me04,PZK98}. The Haar measure on $\mathcal U(d)$ is the unique
	probability measure invariant with respect to left and right multiplication
	with fixed unitary matrices. CUE matrices can be easily sampled starting
	from Ginibre matrices, after performing a QR decomposition \cite{Mez07}.
	\item {\bf random isometry ensemble} consisting of  operators $V : \C^d \to
	\C^D$, for positive integers $d \leq D$, satisfying $V^\dagger V =
	\mathbbm{1}_d$. A Haar-distributed random isometry is obtained by truncating
	a $D \times D$ Haar-distributed random unitary operator $U$ to its first $d$
	columns \cite{ZS00}.
\end{itemize}

\medskip

%%%%%%%%%%%%%%%%%%%%%%%%%%%%%%%%%%%%%%
\section{Distributions of random quantum channels}\label{sec:different-measures}
%%%%%%%%%%%%%%%%%%%%%%%%%%%%%%%%

Let us denote by  $\CC_{d_1,d_2}$ the set of quantum channels acting on density
matrices of order $d_1$ which output density matrices of order $d_2$:
\begin{equation}
\CC_{d_1,d_2} := \{\Phi : M_{d_1}(\C) \to M_{d_2}(\C) \, : \, \Phi \text{ is completely positive and trace preserving}\}.
\end{equation}
The set $\CC_{d_1,d_2}$ is compact and convex. In this section, we discuss
different natural ways of endowing this convex body with natural probability
measures. One natural candidate is the flat measure on this set, induced by
the Hilbert-Schmidt  (HS) distance, $D_{HS}(A,B)=({\rm Tr}
(A-B)(A-B)^{\dagger})^{1/2}$. We shall see that the flat (or Lebesgue) measure
is actually a special case of several one-parameter families of probability
measures on $\CC_{d_1,d_2}$. In the case of equal input and output dimensions
$d_1=d_2=d$, each map $\Phi$ can be represented by a Hermitian matrix $J_\Phi$
of order $d^2$, characterized by $d^4$ real parameters. However, the trace
preserving condition, $[\Tr_{d} \otimes \operatorname{id}_{d}] J_\Phi={\mathbbm
1_d}$, imposes $d^2$ constraints, so  the set $\CC_d$ can be embedded in a real
vector space of dimension $d^4-d^2$.
The volume of the convex set $\CC_d$ with respect to the Hilbert-Schmid (flat)
measure was estimated \cite{SWZ08}  for a large dimension $d$,
while in the case of one-qubit channels, $d=2$,
 an exact result is available \cite{LA17}. 

We introduce next three methods to generate random operations and show their 
equivalence. The motivation for
these families of measures comes from different perspectives on quantum channels
provided by Theorem \ref{thm:def-quantum-channel}. The common idea is that we
shall consider the different representations of a quantum channel (respectively,
the Kraus, Choi, and Stinespring representations), generate randomly the
defining object, (the Kraus operators, the Choi matrix, the Stinespring isometry
respectively), according to the natural measures on the respective spaces, and
then define the probability measure on the set of channels as an induced (or
image) measure. The same strategy has been used in the setting of density
matrices (or mixed quantum states) in \cite{ZSo01}: the \emph{induced measure}
on the set of $d \times d$ density matrices of parameter $s$ is the image
measure of the Lebesgue measure on the unit sphere of $\C^d \otimes \C^s$, by
taking the partial trace on the $s$-dimensional environment, see discussion at
the beginning of Section \ref{sec_distr}.

We shall introduce three families of measures on the set of quantum channels
$\CC_{d_1,d_2}$, starting from the most general ones. We shall conclude by
identifying the flat (or Lebesgue) measure as a special case of all of them.

\bigskip

{\bf a) Random Choi matrix.}

\noindent Define the set of allowed parameters $M$
\begin{equation}
\mathcal M_{d_1,d_2}:= \left\{ \left\lceil \frac{d_1}{d_2}\right\rceil,
\left\lceil \frac{d_1}{d_2}\right\rceil+1, \ldots, d_1d_2-1 \right\} \sqcup
[d_1d_2, +\infty).
\end{equation}

\begin{definition}\label{def:random-Choi} Let $M \in \mathcal M_{d_1,d_2}$ be a
real number. We define $\mu_{d_1, d_2; M}^{Choi}$ to be the probability measure
of the random quantum channel $\Phi \in \CC_{d_1,d_2}$ defined as follows:
\begin{enumerate}
\item Consider a random complex Wishart matrix $W$ of parameters
$(d_1d_2, M)$;
\item Find the positive semidefinite matrix defined by the partial trace,
$H:=[ \Tr_{d_2} \otimes \operatorname{id}_{d_1}] W$;
\item Write the dynamical matrix  (or Choi matrix)
\begin{equation}\label{eq:normalize-Choi}
J:=
({\mathbbm 1}_{d_2} \otimes H^{-1/2}) W ({\mathbbm 1}_{d_2} \otimes H^{-1/2})  \ ;
\end{equation}
\item Reshuffle the Choi matrix $J$ to obtain the superoperator $\Phi=J^{\mathrm
R}$ (see \eqref{eq:def-reshuffling} for the definition of reshuffling); in other
words, $\Phi$ is the unique quantum channel having Choi matrix $J_\Phi = J$.
\end{enumerate}
\end{definition}

Several remarks are in order here. First, note that the condition $M \in
\mathcal M_{d_1,d_2}$ allows for the existence of the Wishart distribution of
$W$. Second, the lower bound on the integer values of $M$, $Md_2 \geq d_1$,
implies that the random matrix $H$ is, generically, invertible. Indeed, $W$
follows a Wishart distribution of parameters $(d_1 d_2, M)$ and thus $H$ is also
Wishart, with parameters $(d_1, Md_2)$. Hence, with probability one, $H$ is
positive definite, rendering valid the normalization procedure from
\eqref{eq:normalize-Choi}. The random matrix $J$ is constructed to be positive
semidefinite, rendering the corresponding channel $\Phi$ completely positive;
the trace preservation condition follows from \eqref{eq:normalize-Choi}. Since
the matrix $H$ is generically invertible, the rank of the Choi matrix $J$ (and
thus the Choi rank of $\Phi$) is, almost surely,
\begin{equation}\label{eq:Choi-rank-random}
\operatorname{rk}_C(\Phi) = \min(d_1d_2, M).
\end{equation}

Finally, let us point out that, from a computational perspective, the costly
operation in the procedure above is the inversion of the $d_1 \times d_1$ matrix
$H$, needed to enforce the trace preservation condition.

\bigskip

{\bf b) Random Kraus operators.}

\begin{definition}\label{def:random-Kraus}
Let $M$ be an integer satisfying $M
d_2 \geq d_1$. We define $\mu_{d_1, d_2; M}^{Kraus}$ to be the probability
measure of the random quantum channel $\Phi \in \CC_{d_1,d_2}$ defined as
follows:
\begin{enumerate}
\item Generate $M$ independent $d_2 \times d_1$ non-Hermitian matrices $G_1, \ldots, G_M $ from
the complex Ginibre ensemble;
\item Compute the positive semidefinite matrix $H=\sum_{i=1}^M  G_i^{\dagger} G_i \geq 0$;
\item Define the set of Kraus operators $A_i:= G_i H^{-1/2}$, $i=1,\dots, M$;
\item The channel $\Phi$ is defined via its Kraus decomposition $\Phi(\cdot) = \sum_{i=1}^M A_i \cdot A_i^\dagger$.
\end{enumerate}
\end{definition}

Let us first justify the validity of the construction. As in the random Choi
matrix setting above, the matrix $H$ has a Wishart distribution of parameters
$(d_1, Md_2)$, hence it is generically positive definite. The operators $A_i$
satisfy the condition
\begin{equation}
   \sum_{i=1}^M A_i^{\dagger} A_i = H^{-1/2} H H^{-1/2} = \mathbbm 1_{d_1},
\end{equation}
proving that the completely positive map $\Phi$ associated with Kraus operators
$A_i$ is trace preserving and forms a legitimate quantum channel. By
construction, the channel $\Phi$ has generically Choi rank given by
\eqref{eq:Choi-rank-random}. This can also be seen as a consequence of the
following result, showing that the probability measures defined as above
correspond to the ones obtained from random Choi matrices, in the case
of integer parameter $M$; for a proof, see Appendix \ref{appendix-1}.

\begin{proposition}\label{prop:Choi-vs-Kraus}
For all integers $M$ such that
$Md_2 \geq d_1$, we have $\mu_{d_1,d_2;M}^{Kraus} = \mu_{d_1,d_2;M}^{Choi}$.
\end{proposition}

Finally, let us point out here that the computational cost of the procedure
presented in Definition \ref{def:random-Kraus} comes from inverting the matrix
$H$.

\bigskip

 {\bf c) Environmental form.}
\begin{definition}\label{def:random-Stinespring} Let $M$ be an integer
satisfying $M d_2 \geq d_1$. We define $\mu_{d_1, d_2; M}^{Stinespring}$ to be
the probability measure of the random quantum channel $\Phi \in \CC_{d_1,d_2}$
defined as follows:
\begin{enumerate}
\item Consider a random Haar isometry $V : \C^{d_1} \to \C^{d_2} \otimes
\C^{M}$ embedding the input system Hilbert space isometrically into the tensor
product of the output space with an environment ${\sf E}$ of dimension $M$;
\item The channel $\Phi$ is defined by its Stinespring decomposition
\begin{equation}\label{eq:def-Stinespring}
\Phi(\cdot)=[\operatorname{id}_{d_2} \otimes \Tr_{M}]\left( V \cdot V^\dagger
\right).
\end{equation}
\end{enumerate}
\end{definition}

We sketch next a construction involving a more physical unitary evolution, which
is however less general than the one above. Consider a total Hilbert space
$\mathcal H$ admitting two tensor product decompositions
\begin{equation}
  \mathcal H = \C^{d_1} \otimes \C^K = \C^{d_2} \otimes \C^M,
\end{equation}
where $K$, respectively $M$ are the dimensions of two auxiliary systems: an
input environment ${\sf E}_{in} = \C^K$, initially in an arbitrary pure state
$|\nu\rangle \in \C^K$, and an output environment ${\sf E_{out}} = \C^M$. Evolve
the total system with a unitary transformation $U$ of size $D= d_1 K = d_2M$,
which is assumed to be generated according to the Haar measure on $\mathcal
U(D)$. The channel $\Phi$ is then defined as
\begin{equation}\label{eq:mu-stinespring-unitary}
  \Phi(\rho)=[\operatorname{id}_{d_2} \otimes
  \Tr_{M}]\left(U(\rho \otimes|\nu\rangle\langle\nu|)U^{\dagger}\right).
\end{equation}
Note that the random isometry $V$ appearing in the first step of construction in
Definition \ref{def:random-Stinespring} can be obtained by truncating the $d_2M
\times d_2 M$ Haar-random unitary operator $U$ to a $d_2 M \times d_1$ matrix
$V$.

As in the previous cases, the corresponding Choi matrix of the channel has
generically rank given by \eqref{eq:Choi-rank-random}. Notice that in this case,
the sampling procedure defined above has as a computational bottleneck, the
sampling of the random Haar isometry $V$. The following result (proven in
Appendix \ref{appendix-2}) shows that environmental form construction, using
random isometries, is also a special case of the random Choi matrix construction
considered previously.

\begin{proposition}\label{prop:Stinesping-equals-Kraus}
	For all integers $M$ such that $Md_2 \geq d_1$, we have $\mu_{d_1,d_2;M}^{Stinespring} = \mu_{d_1,d_2;M}^{Choi}$.
\end{proposition}

\bigskip

 {\bf d) The Lebesgue (flat) measure.}

Finally, a natural probability measure is the (normalized) Lebesgue (or flat)
measure on the set of quantum channels. Since the set $\CC_{d_1, d_2}$ is a
convex compact set, one can endow it with the probability measure obtained by
normalizing the volume (or Hilbert-Schmidt) measure to have total mass 1. We
have the following remarkable statement, see Appendix \ref{appendix-3} for a
proof.

\begin{proposition}\label{prop:Lebesgue-vs-others} The flat (or Lebesgue)
measure $\mu_{d_1,d_2}^{Lebesgue}$ on the set of quantum channels is a
particular case of the constructions in Definitions \ref{def:random-Choi},
\ref{def:random-Kraus}, \ref{def:random-Stinespring}, obtained for the value $M
= d_1d_2$:
\begin{equation}
  \mu_{d_1,d_2}^{Lebesgue} = \mu_{d_1,d_2;d_1d_2}^{Stinespring} = \mu_{d_1,d_2;d_1d_2}^{Kraus} = \mu_{d_1,d_2;d_1d_2}^{Choi}.
\end{equation}
\end{proposition}

\bigskip

To conclude, we have provided several classes of probability measures on the set
of quantum channels $\CC_{d_1,d_2}$, indexed by a real or integer parameter $M$,
which coincide for identical values of $M \geq d_1 / d_2$. The proof of 
the equivalence of measures generated by families \textbf{a)} and \textbf{b)} 
is the consequence of the isomorphism defined in 
\eqref{jamiol}, while the equivalence with \textbf{c)} follows from the fact 
that Ginibre ensemble induces Haar measure on the set of unitary matrices, 
given the transformation $G \mapsto G (G^\dagger G)^{-1/2}$.
 
 The proposed families can be ordered, from particular to general, as below 
 (see Propositions \ref{prop:Choi-vs-Kraus},
\ref{prop:Stinesping-equals-Kraus},\ref{prop:Lebesgue-vs-others}). The following
relation is the main result of the first part of this work:
\begin{equation}\label{eq:all-measures}
\mu_{d_1,d_2}^{Lebesgue} \in
\bigg\{ \mu_{d_1,d_2;M}^{Stinespring} \bigg\}_{\substack{M \in \mathbb N \\ M \geq d_2/d_1}} =
\bigg\{ \mu_{d_1,d_2;M}^{Kraus} \bigg\}_{\substack{M \in \mathbb N \\ M \geq d_2/d_1}} \subset
\bigg\{ \mu_{d_1,d_2;M}^{Choi} \bigg\}_{M \in \mathcal M_{d_1,d_2}}.
\end{equation}

Note that each of the above procedures has its advantages. The environmental
form \textbf{c)} has a clear physical interpretation and can be approximated in
an experiment, in which a random unitary matrix $U$ can be approximated as an
evolution operator of a quantum chaotic system \cite{Ha10,Br01}. On the other
hand it is not suitable for numerical simulations. To see this, let us consider,
for the sake of simplicity, the case $d_1 = d_2 = d$, corresponding to the same
input and output system sizes. In order to obtain a distribution, parameterized
by $M$ on the set of quantum channels transforming $d$ dimensional systems to
$d$ dimensional systems, we need to generate and store a unitary matrix $U \in
\mathcal U(M d)$. This in turn involves computing the QR decomposition, which
for a matrix of dimension $n$ has the complexity $O(n^3)$. In our case we get at
least $O(M^3 d^3)$ multiplications. For $M=d^2$ we get the complexity $O(d^9)$.
The forms \textbf{a)} and \textbf{b)} based on Wishart matrices and independent
random Kraus operators respectively are the easiest to work with in numerical
simulations of a generic quantum channel. Both cases \textbf{a)} and \textbf{b)}
involve calculating the inverse of square of a $d$ dimensional matrix. Aside
from this we have in \textbf{b)} $d^2$ multiplications of $d$ dimensional
matrices. As matrix multiplication in typical implementations has the complexity
of $O(d^3)$, we get that the overall complexity is $O(d^5)$. The case
\textbf{a)} involves the multiplication of $d^2$ dimensional matrix $W$ hence it
has the complexity of at least $O(d^6)$; note also that there does not exist a
simple procedure to sample from a Wishart distribution of parameters $(d_1d_2,
M)$ for non-integer $M \in \mathcal M_{d_1,d_2}$. Therefore, for numerical
implementations one can recommend algorithm \textbf{b)} involving random Kraus
operators.

Several other families of probability distributions on the set $\mathcal
C_{d_1,d_2}$ of physical and mathematical interest are discussed in Appendix
\ref{app:other-measures}.

\section{Distribution of output states of random quantum channels}
\label{sec_distr}

We consider in this section the output state of a random quantum channel, for a
given input. We start by recalling the induced measures on the set of density
matrices. This one-parameter family of probability measures $\nu_{d;s}$ has been
introduced in \cite{SZ04} and can be described in two equivalent ways. Let us
define, for a given Hilbert space dimension $d$, the set of admissible
parameters
\begin{equation}
\mathcal S:=\{1,2, \ldots, d-1\} \sqcup [d, \infty),
\end{equation}
which is precisely the set of allowed parameters for the complex Wishart
distribution. On the one hand, one can consider a complex Wishart matrix $W$ of
parameters $(d,s)$ (where $d$ is the size of $W$ and $s \in \mathcal S$ is a
parameter) and normalize its trace:
\begin{equation}
\frac{W}{\Tr W} \sim \nu_{d;s}.
\end{equation}
Equivalently, for the integers $s \in \mathcal S$, one can consider a uniformly
distributed vector $x$ on the unit sphere of $\C^{ds}$ and take its partial
trace with respect to the ``environment'' $\C^s$:
\begin{equation}\label{eq:induced-measure-environment}
	[\mathrm{id}_d \otimes \Tr_s] \ketbra x x \sim \nu_{d;s}.
\end{equation}

Remarkably, the uniform measure on the set of $d \times d$ density matrices
corresponds to the particular value $s=d$: $\nu^{Lebesgue} = \nu_{d;d}$
\cite{ZSo01}. This fact is to be compared with the situation for quantum
channels, see Proposition \ref{prop:Lebesgue-vs-others} and
Eq.~\eqref{eq:all-measures}.

\begin{proposition}\label{prop:output-state-induced-measure}
	Let $\Phi:M_{d_1}(\C) \to M_{d_2}(\C)$ be a random quantum channel having
	distribution $\mu^{Stinespring}_{d_1,d_2;M}$ for integer $M \geq d_1/d_2$.
	Then, for any given fixed
	pure input state $\ketbra \psi \psi$, the output state $\Phi(\ketbra \psi
	\psi)$	has distribution $\nu_{d_2;M}$.
\end{proposition}

\begin{proof}
	We have, using the Stinespring form of the random quantum channel $\Phi$:
	\begin{equation}\Phi(\ketbra \psi \psi) = [\mathrm{id}_{d_2} \otimes \Tr_M]
	(V\ketbra{\psi}{\psi}V^*) = [\mathrm{id}_{d_2} \otimes \Tr_M] \ketbra{x}{x},
	\end{equation}
	where $|x \rangle := V\ket{\psi}$. Since the isometry $V$ is
	Haar-distributed and the unit vector $\psi$ is fixed, the vector $x$ is
	uniformly distributed on the unit sphere of $\C^{d_2M}$. The conclusion
	follows from the environmental description of the induced measures, see
	Eq.~\eqref{eq:induced-measure-environment}.
\end{proof}

From the proposition above, we can infer that the average of the output state
(with respect to the randomness in the channel) for a fixed input is the
maximally mixed state
\begin{equation}\E \Phi(\ketbra \psi \psi) = \E_{\nu_{d_2;M}} \rho =
\frac{\1_{d_2}}{d_2}.
\end{equation}
This fact is equally a consequence of the following result.

\begin{proposition}
	The average of a random quantum channel having distribution
	$\mu^{Stinespring}_{d_1,d_2;M}$ is the maximally depolarizing channel $\E
	\Phi = \Phi_*$, with
	\begin{equation}
	\begin{split}
		\Phi_*: M_{d_1}(\C) &\to M_{d_2}(\C)\\
		X &\mapsto \Tr(X) \frac{\1_{d_2}}{d_2}.
		\end{split}
	\end{equation}
\end{proposition}
\begin{proof}
	The conclusion follows easily from the computation of the average Choi
	matrix $J_\Phi$ using Weingarten calculus \cite{cs06,fkn19}:
	\begin{equation}\E J_\Phi = M \1_{d_2d_1} \frac{1}{d_2M} =
	\frac{\1_{d_2d_1}}{d_2} = J_{\Phi_*}.
	\end{equation}
\end{proof}

Let us now consider two statistical quantities associated to an arbitrary
quantum channel: the \emph{average output purity} and the \emph{unitarity}
\cite{wal19}:
\begin{equation}
\begin{split}
p(\Phi) &= \E \Tr \left(\Phi(\ketbra \psi \psi)^2 \right)\\
u(\Phi) &= \frac{d_1}{d_1-1}\E \Tr \left( (\Phi(\ketbra \psi \psi) -
\Phi(\1_{d_1}/d_1))^2 \right),
\end{split}
\end{equation}
where the expectation corresponds to the choice of a uniform unit vector $\psi$
on the unit sphere of the input space $\C^{d_1}$. Note that for unital channels
(satisfying $\Phi(\1_{d}) = \1_{d}$ for $d=d_1=d_2$), the two quantities above
are related by the relation $\frac{d-1}{d} u(\Phi)=p(\Phi)-1/d$. We compute the
averages of these two quantities in the next proposition.

\begin{proposition}\label{prop:purity-unitarity}
	Let $\Phi:M_{d_1}(\C) \to M_{d_2}(\C)$ be a random quantum channel having
	distribution $\mu^{Stinespring}_{d_1,d_2;M}$, for an integer $M \geq
	d_1/d_2$. Then
	  the expectation values of the average output purity and unitarity read,
	\begin{align}
		\label{eq:average-output-purity}\E p(\Phi)& = \frac{d_2+M}{d_2M+1},\\
		\label{eq:average-unitarity}\E u(\Phi)& = \frac{M(d_2^2-1)}{(d_2M)^2-1}.
	\end{align}
\end{proposition}
\begin{proof}
	First of all, note that the expectation over the random pure state
	$\ket{\psi}$ in the definition of the quantities $p,u$ can be absorbed in
	the expectation over the random channel $\Phi$. Hence, in the following, we
	shall assume that $\ket{\psi}$ is some fixed unit vector in $\C^{d_1}$. We
	shall make use of the following result, proven in Lemma~\ref{lem:Wg} in
	Appendix \ref{app:lemma-Wg}:
	\begin{equation}\E \Tr[\Phi(A)\Phi(B)] = \frac{(\Tr A)(\Tr B)d_2(M^2-1) +
	\Tr(AB)M(d_2^2-1)}{(d_2M)^2-1}.
	\end{equation}

	Applying the result above for $A=B =\ketbra \psi \psi$ gives us
	Eq.~\eqref{eq:average-output-purity}. Note that this result could have been
	obtained directly from Proposition \ref{prop:output-state-induced-measure},
	using the formula for the average purity of a random density matrix from
	\cite[Eq.~(5.11)]{SZ04}.

	To show \eqref{eq:average-unitarity}, we make use again of Lemma
	\ref{lem:Wg}, this time for $A=\ketbra \psi \psi$, $B=\1_{d_1}/{d_1}$ and
	then with $A=B=\1_{d_1}/{d_1}$.
\end{proof}
Note that in the regime where $d_2 \to \infty$, the average purity of a random
quantum channels scales as $1/M$.

\begin{corollary}
	The average output purity and the average unitarity of a uniformly
	distributed random quantum channel $\Phi \sim \mu^{Lebesgue}_{d,d}$ are

\begin{equation}
\E p(\Phi) = \frac{d}{d^2-d+1} \qquad \text{and}\qquad \E u(\Phi) =
\frac{d^2}{d^4+d^2+1}.
\label{unitarit}
\end{equation}
\end{corollary}

\section{Spectral properties of the superoperator}
\label{sec:spectrum-random-quantum-channels}

In this section we analyze spectral properties (singular values and eigenvalues)
of generic superoperators $\Phi$ represented by a non-Hermitian matrix of order
$d^2$ (we consider the case $d_1=d_2=d$ here). Note that we use the same letter
$\Phi$ to denote the linear map representing a quantum channel $\Phi : \mathcal
M_d(\mathbb C) \to \mathcal M_d(\mathbb C)$ and the corresponding matrix, when
seen as an operator on $\mathbb C^{d^2} \cong \mathcal M_d(\mathbb C)$. As
explained in Section \ref{sec:different-measures}, one obtains the superoperator
matrix by reshuffling the Choi (or dynamical) matrix, $\Phi = J_{\Phi}^R$.

Before we move on to study random superoperators, let us first recall some
general properties of such matrices. If $\Phi$ is the superoperator of a quantum
channel (completely positive and trace preserving linear map), then
\cite{ehk,gro}:
\begin{enumerate}
\item the spectrum of $\Phi$ is contained in the unit disk $\{z \in \mathbb C 
\, : \, |z| \leq 1\}$,
\item there exists a Perron-Frobenius eigenvalue $\lambda_1=1$,
\item the eigenspace of the eigenvalue $\lambda_1=1$ contains a positive semidefinite element.
\end{enumerate}
Actually, the structure of the spectrum of a superoperator is much richer (see,
e.g.,~\cite{bna,wol}), but the properties above are the only ones we need in
this paper. Of a crucial importance is the modulus of the sub-leading eigenvalue
$r=|\lambda_2|\leq 1$ and the spectral gap  $\gamma :=1-r \geq 0$ which
determines the convergence of the system to the equilibrium, see
Fig.~\ref{fig:2}.

\begin{figure}[h]%disk
	\includegraphics[angle=0,width=0.7\columnwidth]{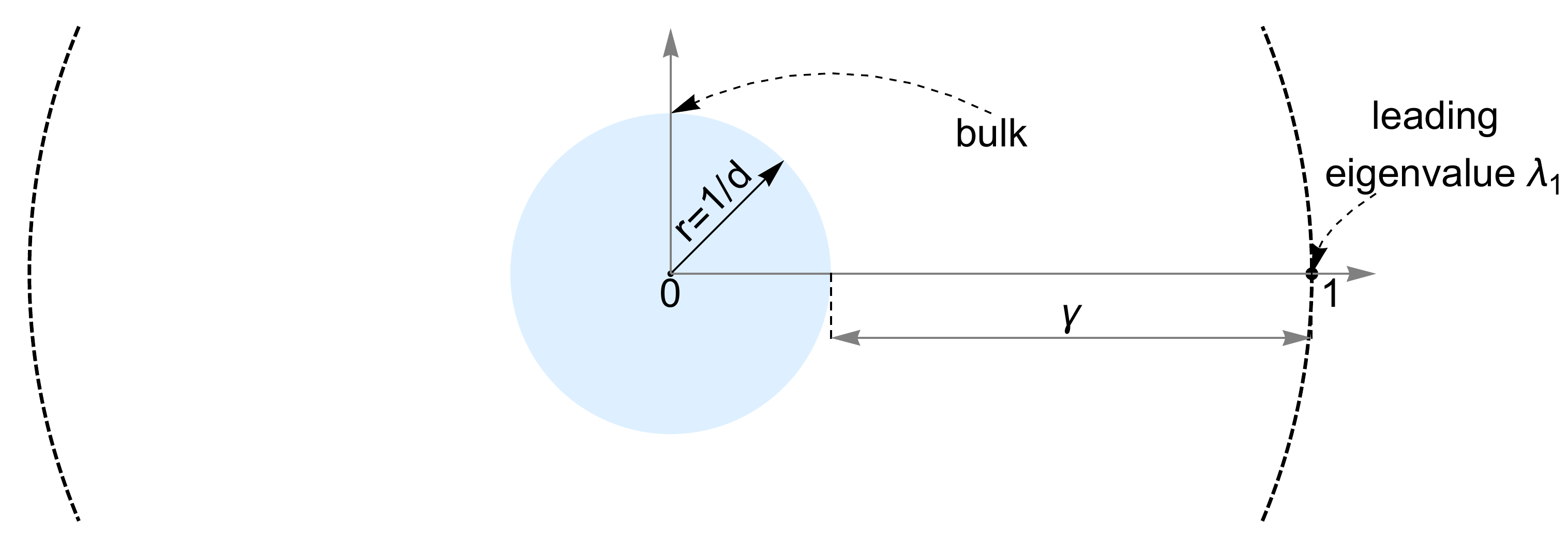}
	\caption{Sketch of the spectrum of a superoperator $\Phi$
	   associated with a  random channel, which consists of the
	      leading Perron-Frobenius eigenvalue $\lambda_1=1$ and a bulk
	     forming the Girko disk of radius $r\approx 1/d$. The spectral gap
	     reads $\gamma=1-r$.
	     }
	\label{fig:2}
\end{figure}

It was noted in \cite{BCSZ09} that the properties of the superoperator $\Phi$
corresponding to a random operation can be modeled by the real Ginibre ensemble.
To explain this fact it is convenient to use the Bloch vector representation of a map.
Any density operator $\rho$ of size $d$ can be represented
using the generalized Bloch vector,
\begin{equation}
\rho=\frac{1}{d} \left(
\mathbbm 1_d + \sum_{i = 1}^{d^2-1} \tau_i \; \Lambda_i \
\right),
\label{Bloch1}
\end{equation}
where  $\Lambda_i$ denotes the vector of three Pauli matrices  for $d=2$, and
are proportional to eight Gell-Mann matrices for $d=3$, while for a higher
dimension it represents the vector of $d^2-1$ hermitian and traceless generators
of $\textsf{SU(}d\textsf{)}$, normalized as ${\rm Tr}\left(\Lambda_i \Lambda_j
\right)= d \, \delta_{ij}$. Usually the order of the generators is not relevant,
but for the purpose of studying the quantum to classical transition and the
effects of super-decoherence and coherification of a channel
\cite{KCPZ18,KCPZ19} it will be convenient to choose the order $\sigma_z,
\sigma_x, \sigma_y$, and in higher dimensions select first  $d-1$ generators
$\Lambda_i$ as diagonal ones -- see Section \ref{sec:classical}. Since any
density matrix $\rho$  is Hermitian, all components of the Bloch vector ${\vec
\tau}$ are real, $\tau_i= {\rm Tr}\Lambda_i \rho  \in\mathbb{R}$ for
$i=1,\dots,d^2-1$.

In the case of a state $\rho_{AB}$ of a a bipartite $d \times d$ system is its
convenient to expand the density matrix in the product basis formed by tensor
products of the generators,  $\Lambda_i \otimes \Lambda_j$. It leads to the
following {\sl Fano form} \cite{Fa83} applicable to any bi-partite state
\cite{BZ17},
 \begin{equation}
\rho_{AB} =   \frac{1}{d^2} \sum_{i,j = 0}^{d^2-1} {\tilde R}_{ij} \Lambda_i \otimes \Lambda_j .
\label{Fano1}
\end{equation}
The expansion coefficients are given by the projection of the state onto
the elements of the product basis,
 ${\tilde R}_{ij}= \tr  \left( \rho_{AB} \left(\Lambda_i  \otimes  \Lambda_j
 \right) \right)$.
 As $\Lambda_0= {\mathbbm 1}$ the matrix ${\tilde R}$ takes the form
 \begin{equation}
{\tilde R} \; =\; \left[
\begin{array}{ll}
  1 & \vec a^T \\
{\vec b}&{ R}
\end{array}
\right] .
\label{fano2}
\end{equation}

The vectors $\vec a$ and $\vec b$ of length $d^2-1$ represent Bloch vectors
$\tau_A$ and $\tau_B$ of both partial traces, $\rho_A=[\mathrm{id}_A \otimes
{\rm Tr}_B] \rho_{AB}$ and  $\rho_B=[{\rm Tr}_A \otimes \mathrm{id}_B ]
\rho_{AB}$, respectively. These parameters can be thus determined locally, while
the square matrix $R$ of size $d^2-1$, a truncation of ${\tilde R}$, describes
correlations between both subsystems. In the case of a product state,
$\rho_{AB}=\rho_A \otimes \rho_B$, its elements are $R_{ij}=a_ib_j$,  so the
state is separable. In general, the real correlation matrix $R$ is non
symmetric, and for a fixed local Bloch vectors $\vec a$ and $\vec b$ only a
suitable choice of $R$ assures positivity of the state \cite{BH300}. The norm of
the correlation matrix $R$ can be used to formulate separability criteria -- for
a sufficiently small norm $||R||$ the state  $\rho_{AB}$ is separable
\cite{Vi07,H409}.

 Let $\vec \tau$ represents an initial state $\rho$ and $\vec \tau'$ be the
 Bloch vector of the image $\rho'=\Phi(\rho)$, were for simplicity we assumed
 that both dimensions are equal,  $d_1=d_2=d$. Any channel $\Phi$ can be
 now represented by the action on the Bloch vector,
\begin{equation}
{\vec \tau}' =   Q {\vec \tau} + {\vec \kappa},
\label{Bloch2}
\end{equation}
 where $Q$ is a real matrix of size $d^2-1$, while $\vec \kappa$
 is a translation vector of length $d^2-1$, which vanishes for unital maps.
Hence the superoperator $\Phi$ can be represented  \cite{TDV00}
 by an asymmetric  real matrix of order $d^2$,
\begin{equation}
\tilde{\Phi} \; =\; \left[
\begin{array}{ll}
  1 & 0 \\
{\vec \kappa}&{ Q}
\end{array}
\right] .
\label{super2}
\end{equation}
The above form is convenient to for spectral analysis: the spectrum of the
superoperator $\Phi$ consists of the leading eigenvalue $\lambda_1=1$ and the
$d^2-1$ eigenvalues of the matrix $Q$, which can be complex.
The trace of the real distortion matrix $Q$ has an operational 
interpretation as it determines the average fidelity between a random
pure state $|\psi\rangle \langle \psi|$ and its image
with respect to map $\Phi$ \cite{KBF20}.

Note a similarity  between the form (\ref{super2}) of an arbitrary operation
$\Phi$ and the matrix  (\ref{fano2}) appearing in the Fano form of a bipartite
state. The vector $a$ vanishes here due to the trace preserving condition. The
observed analogy can be formally explained with use of the Jamio{\l}kowski
isomorphism.

\begin{proposition}  \label{prop:Fano3}
  Bloch representation (\ref{super2}) of a quantum operation $\Phi$
  is equivalent with the Fano form (\ref{fano2})
  of the partially transposed Jamio{\l}kowski state~(\ref{jamiol}), %i.e.
  $J_{\Phi}^{T_2}/d$.
\end{proposition}

\begin{proof}
Action of the map $\Phi$ on an arbitrary operator $X$ can be expressed by the
corresponding Choi matrix $J_{\Phi}$ in the following way \cite[Chapter
2.2]{Wat18}

\begin{equation}\label{eqn:choi-action}
\Phi(X) = [\mathrm{id}_{d} \otimes \Tr_{d}] \left( J_{\Phi} (\1_{d}
\otimes
X^T) \right).
\end{equation}
Matrix elements of the Bloch representation  (\ref{super2}) of a channel $\Phi$,
 can be written as $\tilde{\Phi}_{ij} = \frac1d \tr(\Lambda_i \Phi(\Lambda_j))
 $. Therefore making use of  formula~\eqref{eqn:choi-action} we get
\begin{equation}
\begin{split}
\tilde{\Phi}_{ij} &= \frac1d \tr (\Lambda_i \Phi(\Lambda_j)) =
\frac1d \tr \left(\Lambda_i [\mathrm{id}_{d} \otimes \Tr_{d}]\left(J_{\Phi}
(\1
\otimes \Lambda_j^T)\right)\right)
=\frac1d \tr \left((\Lambda_i \otimes \1) J_{\Phi} (\1 \otimes
\Lambda_j^T)\right)\\
&=\frac1d \tr \left( J_{\Phi} (\Lambda_i \otimes \Lambda_j^T) \right)
= \frac1d \tr \left(J_{\Phi}^{T_2} (\Lambda_i \otimes \Lambda_j) \right)
=\tilde{R} \left( J_{\Phi}^{T_2}/d \right)_{ij}.
\end{split}
\end{equation}
\end{proof}

Alternative proof of the above fact, using Kraus representation, can be
formulated using algebraic Lemma~\ref{lemma:ACB^TD^T} stated in
Appendix~\ref{app:lamma}.

Proposition  \ref{prop:Fano3} shows a direct link between the problem of finding
restrictions for the correlation matrix $R$ to assure positivity of the
bipartite state $\rho_{AB}$ in (\ref{Fano1}) analyzed in \cite{BH300} and the
question, for what  matrix  $Q$ appearing in (\ref{Bloch2}) the corresponding
map $\Phi$ is completely positive \cite{FA99, RSW02}. For any trace preserving
quantum operation the vector $a$ in (\ref{fano2}) vanishes, so the second
question can be considered as a special case of the first one. In the simplest
case of one-qubit operation, $d=2$, conditions for the matrix  $Q$ which imply
complete positivity are known  \cite{FA99, RSW02, BGNPZ14}, but for larger
dimensions the problem becomes rather complicated.

However,  if a random channel $\Phi$ of a large dimension $d\gg 1$ is generated
according to the flat measure, $\mu_{d,d}^{Lebesgue}$, these constraints become
weaker. In the limit of large dimension $d$ the first two cumulants of $\Phi$
are identical to those of the Gaussian distribution and the higher cumulants can
be neglected \cite{BSCSZ10}. This yields an evidence that the statistical
properties of  the matrix $Q$ of order $d^2-1$, forming the core of the
superoperator $\Phi$, can be described by a random matrix  $G_\R$ from the real
Ginibre ensemble. Thus the spectrum of $Q$  forms in the complex plane a scaled
\emph{Girko disk} \cite{Gi84}. Normalization of the Ginibre matrix implies that
the eigenvalues concentrate in the disk of radius $1/\sqrt{M}$, where $M$ is the
number of random Kraus operators defining the map  \cite{BCSZ09}.

In the case $M=d^2$, corresponding to the uniform distribution of channels, the
radius of the disk behaves as $r\approx 1/d$ and for large dimension the
distribution in the disk becomes uniform. Since the trace preserving condition
assures the leading Perron-Frobenius eigenvalue $\lambda_1=1$, the average size
of the spectral gap,  $\gamma=\lambda_1 -|\lambda_2|$, behaves as  $\langle
\gamma\rangle  \approx 1-1/d$.

Let us also mention that the average number of real eigenvalues of a uniform
random superoperator $\Phi$ of size $d^2$ has been numerically observed
\cite{BSCSZ10} to behave asymptotically as $\sqrt{2/\pi}\cdot d$, fitting the
corresponding fraction for the real Ginibre ensemble \cite{eko}. In particular,
for large dimensions, the fraction of real eigenvalues decreases as $1/d$.

Rigorous mathematical results on the spectral gap have been obtained in the
setting where the number $M$ of Kraus operators is fixed, in relation to the
so-called \emph{quantum expanders}. Hastings started this line of inquiry in
\cite{has07}, where he showed that random mixed unitary channels with $M$
unitary Kraus operators have a spectral gap of order $\gamma \approx 1-1/{\sqrt
M}$. Similar results were obtained for random quantum channels coming from
random isometries in \cite{ggn,lpg}.

\bigskip

Let us now discuss singular values of a random superoperator $\Phi$.
 The leading singular value is equal to unity,
 which is a consequence of preservation of trace,
\begin{equation}
  \langle \chi | \Phi | \chi \rangle =
  \langle \chi |  \chi \rangle = 1,
\end{equation}
where $|  \chi \rangle = | \rho_{\rm inv} \rangle \rangle / \| \rho_{\rm inv}
\|_2$ represents the normalized vector
of length $d^2$
corresponding to the invariant state of the map, $\rho_{\rm inv}=
\Phi(\rho_{\rm inv})$. To analyze the remaining singular values, we draw again
a parallel between the matrix $\Phi$ and a properly normalized reshuffled
Wishart matrix of parameters $(d^2,M)$, $W^R/(dM)$. In the case of Wishart
matrices, it has been shown in \cite[Theorem 3.1]{ane} that, in the regime
where $d,M \to \infty$, the singular values of
\begin{equation}
  \sqrt M\left[ \frac{W^R}{dM} - |\chi \rangle \langle  \chi | \right]
\end{equation}
converge, in moments, to the quarter-circle law,
\begin{equation}
  \mathrm{d}\mu_{qc} = \frac{\sqrt{4-x^2}}{\pi} \mathbf{1}_{[0,2]}(x) \mathrm{d}x,
\end{equation}
related to the Mar\v{c}henko-Pastur distribution. This is evidence towards the
claim that, once the projection on the Perron-Frobenius eigenvalue is removed,
the spectral norm of the remaining matrix is of order $1/\sqrt M$.

\section{Classical maps: random stochastic matrices}
\label{sec:classical}

In this section we describe how classical objects (probability vectors and
stochastic maps) arise from quantum objects (density matrices and quantum
channels) and compare their probability distributions.

A quantum state $\rho$ subjected to the coarse-graining map, which describes
quantum decoherence, produces a classical probability vector, $p={\rm diag}
(\rho)$. In the next proposition we compute the probability distribution on the
probability simplex induced by the induced measures on the set of density
matrices. Let us recall first that a probability vector $p=(p_1, \ldots, p_d)
\in \Delta_d$ has a Dirichlet distribution with parameter $s$, written $p \sim
\operatorname{D_s}(p_1,\ldots,p_d)$, if it has density reads
\begin{equation}
  \operatorname{D_s}(p_1,\ldots,p_d)  = C \;
  p_1^{s-1} p_2^{s-1} \cdots p_{d-1}^{s-1}(\underbrace{1-p_1-\cdots -p_{d-1}}_{p_d})^{s-1},
\end{equation}
with a suitable normalization constant $C$.
We refer the reader to \cite[Chapter XI.4]{dev86} for more details on Dirichlet distributions.

\begin{proposition}\label{prop:diag-random-density-matrix}
	Let $\rho \in M_d(\C)$ be a random density matrix from the induced ensemble of
	parameters $(d,s)$: $\rho = [\mathrm{id}_d \otimes
	\mathrm{Tr}_s](\ketbra{\psi}{\psi})$, where $\psi \in \C^{ds}$ is uniformly
	distributed on the unit sphere. Then, $p = \operatorname{diag}(\rho) \in
	\Delta_d$ has a Dirichlet distribution $\operatorname{D_s}(p_1,\ldots,p_d)$ on the
	probability simplex.
\end{proposition}

\begin{proof}
	Recall from \cite{ZSo01,ZPNC11} random density matrices from the induced
	ensemble of parameters $(d,s)$ are obtained as
  \begin{equation}
    \rho = \frac{GG^\dagger}{\operatorname{Tr}(GG^\dagger)},
  \end{equation}
	where $G \in M_{d \times s}(\C)$ is a random rectangular Ginibre matrix
	(having i.i.d.~standard complex Gaussian entries). In particular, the
	diagonal entries read
  \begin{equation}
    p_i = \rho_{ii} = \frac{h_i}{\sum_{j=1}^d h_j},
  \end{equation}
  with $h_i = \sum_{k=1}^s |G_{ik}|^2$. It is easy to see that $2h_i$ has a
	chi-squared distribution of parameter $2s$ (taking into account the real and
	imaginary parts of $G_{ik}$). Hence, $h_i$ are i.i.d.~random variables with
	distribution $\operatorname{Gamma}(s)$. The conclusion follows now from the
	fact that normalizing i.i.d.~Gamma random variables yields the Dirichlet
	distribution, see \cite[Theorem XI.4.1]{dev86}.
\end{proof}

\begin{remark}
	The uniform distribution $\operatorname{D_1}(p_1,\ldots,p_d)$ on the
	probability simplex is obtained by taking the diagonal of random pure states
	($s=1$). The diagonal of uniformly distributed random density matrices
	($s=d$) does \emph{not} yield uniform probability vectors, but the
	distribution $\operatorname{D_d}(p_1,\ldots,p_d)$ which is more concentrated
	towards the ``central'' point $(1/d, \ldots, 1/d)$ of the probability
	simplex $\Delta_d$. The same effect occurs for any $s>1$, see
	Fig.~\ref{fig:simplex-Dirichlet}.
\end{remark}

\begin{figure}
	\centering
	\includegraphics[width=0.4\linewidth]{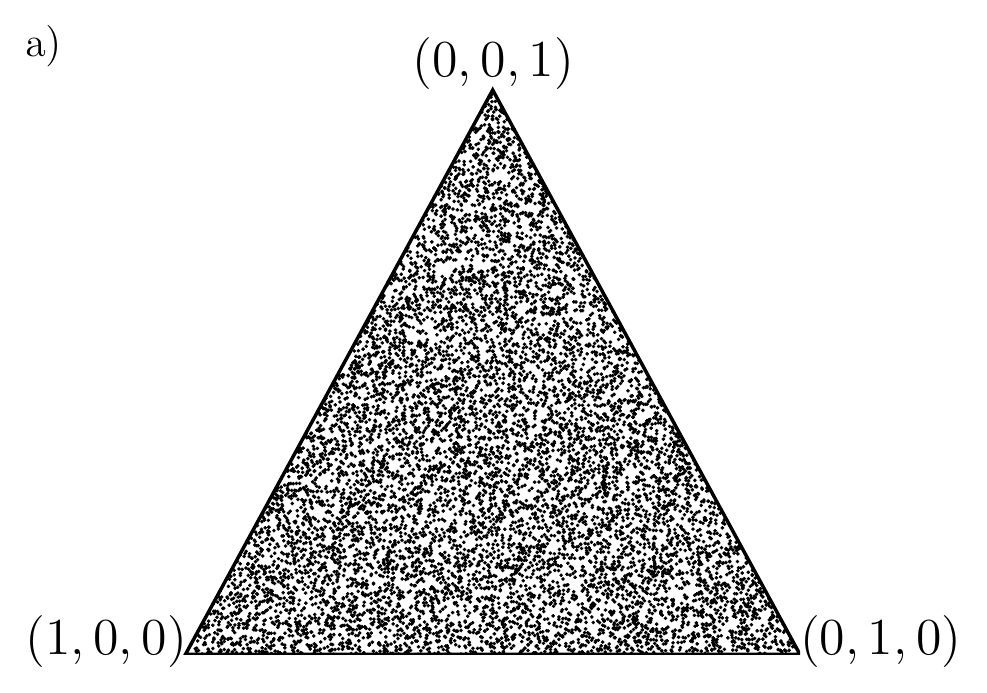}
	\includegraphics[width=0.4\linewidth]{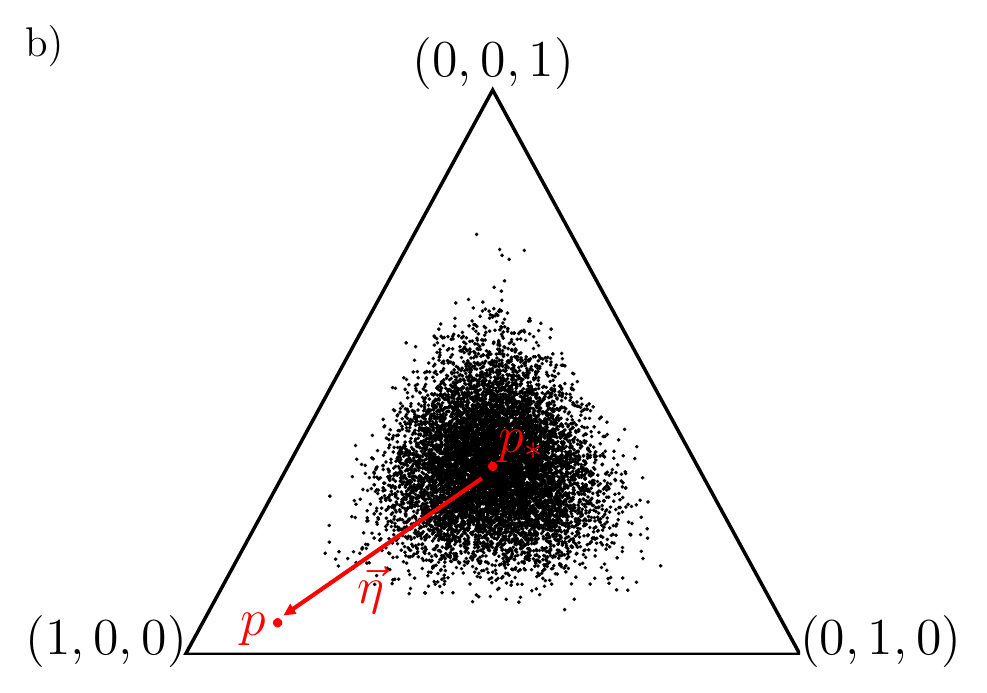}\\
	\includegraphics[width=0.25\linewidth]{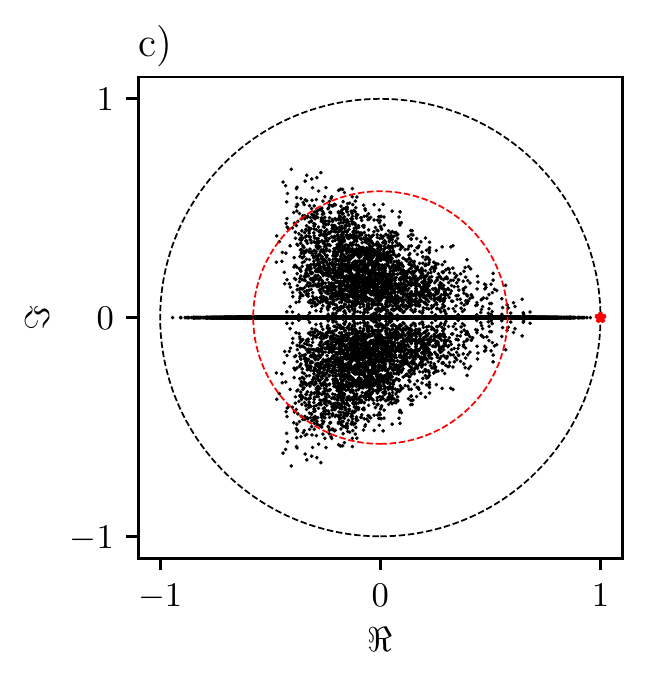}
	\includegraphics[width=0.25\linewidth]{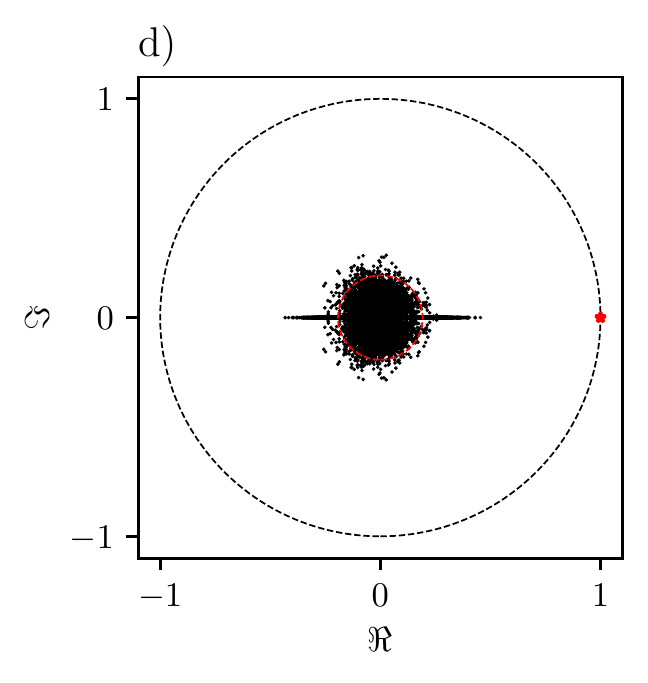}
	\includegraphics[width=0.25\linewidth]{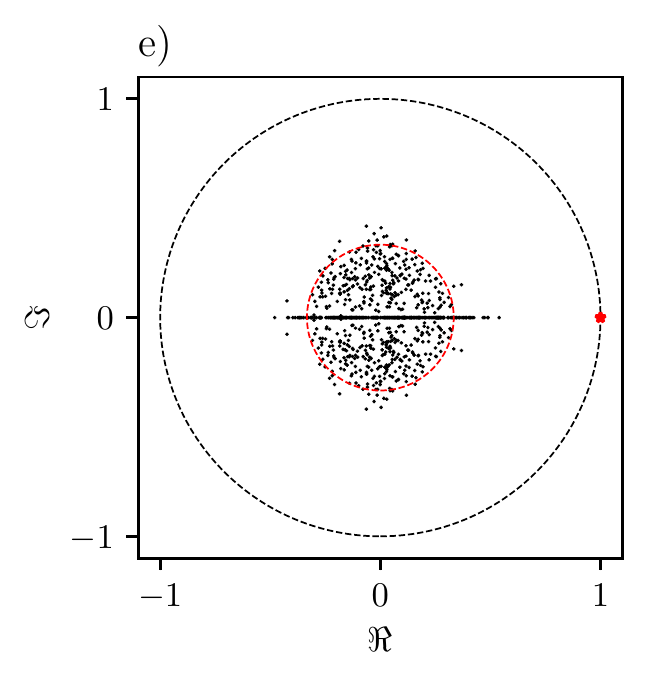}
	\caption{Dirichlet distributions ($10^5$ samples) on $\Delta_3$, for
	a) $s=1$ (left, uniform distribution on the simplex) and
	b)  $s=9$ (right)
	with an exemplary translation vector ${\vec\eta}$,
    a projection of the shift	$\delta=p-p_*$
    on the $(d-1)$ dimensional simplex $\Delta_d$,
    used in Eq. (\ref{Bloch_class}, \ref{stoch2}). Panels c) and d) show
	eigenvalues of $10^3$ stochastic matrices of order $d=3$ with columns sampled
	from $D_1(T_{1j}, T_{2j}, T_{3j})$ and $D_9(T_{1j}, T_{2j}, T_{3j})$
	respectively, corresponding to panels a) and b). Red dashed circles have
	radii $r_1=1/\sqrt{3}$ and $r_c=1/3\sqrt{3}$. Panel e) shows spectra
	of $10^3$ random quantum channels acting on $M_3$
	sampled with the Lebesgue measure, the red dashed circle has radius
	$r_q=1/3$.}
	\label{fig:simplex-Dirichlet}
\end{figure}

\bigskip

In the same way that one obtains a classical probability distribution from a
quantum state, one can obtain a classical Markov map from a quantum channel. Any
quantum channel $\Phi$  generates a corresponding classical transition matrix
$T$, in two equivalent ways. For any  channel $\Phi : M_{d_1}(\C) \to
M_{d_2}(\C)$, we associate the column-stochastic matrix
\begin{equation}
   M_{d_2 \times d_1}(\R) \ni  T_{ji} = \langle j | \Phi(\ketbra{i}{i}) | j
   \rangle.
\end{equation}
Equivalently, $T$ can be obtained by reshaping the diagonal of the Choi matrix
$J_\Phi=\Phi^R$ of size $d_1d_2$ into a matrix of size $d_2 \times d_1$. As
matrix $J$ is weakly positive, so are all entries of $T$. It is easy to check
that if $\Phi$ satisfies the trace preserving condition $[\Tr_{d_2} \otimes
\operatorname{id}_{d_1}] J_\Phi={\mathbbm 1}_{d_1}$ the matrix $T$ is
(column-)stochastic, and if $\Phi$ is also unital ($d=d_1=d_2$),
$[\operatorname{id}_{d} \otimes \Tr_{d}] J_\Phi={\mathbbm 1_d}$, the
corresponding matrix $T$ is also bistochastic -- see e.g. \cite{KCPZ18}.  If the
quantum map $\Phi$ is written in its Kraus form, $\Phi(\rho)=\sum_{k=1}^M A_k
\rho A_k^{\dagger}$, then the quantum superoperator is represented by a matrix
of size $d_1d_2$ obtained as a sum of Kronecker products $\Phi=\sum_{k=1}^M A_k
\otimes {\bar A_k}$, while the corresponding classical transition matrix of size
$d_2 \times d_1$ can be represented by a sum of  Hadamard products
$T=\sum_{k=1}^M A_k \odot {\bar A_k}$. Observe that the trace preserving
condition $\sum_{k=1}^M A_k^{\dagger} A_k ={\mathbbm 1}_{d_1}$ implies the
stochasticity condition $\sum_{j=1}^{d_2} T_{ji}=1$, while the dual unitality
condition,   $\sum_{k=1}^M A_kA_k^{\dagger} ={\mathbbm 1}_{d_2}$, implies that
also its transpose $T^\top$ is (column-)stochastic, and thus $T$ is
bistochastic.

We have shown above that any quantum operation $\Phi$ determines a certain
classical transition matrix, which can appear   due the the effects of
super-decoherence \cite{KCPZ18}. In the four--index notation one can write,
$T_{ji}=\Phi_{j j, i i }=(J_\Phi)_{j i, j i }$. Thus any ensemble of random
quantum operations induces a certain ensemble of classical stochastic matrices.
Let us analyze this distribution for the ensemble of quantum channels from
Section \ref{sec:different-measures} defined by the environmental form
\textbf{c)} and random Haar isometries.

\begin{proposition}\label{prop:T-from-Phi} Let $\Phi : M_{d_1}(\C) \to
	M_{d_2}(\C)$ be a random quantum channel from the ensemble
	$\mu^{Stinespring}_{d_1,d_2;M}$, where $M \geq d_1/d_2$ is a fixed
	parameter. Consider the induced measure on column-stochastic maps $T \in
	M_{d_2 \times d_1}(\R)$. Then, every column of $T$ has Dirichlet
	distribution with parameter $M$ on the simplex $\Delta_{d_2}$. However, the
	columns of $T$ are not independent, the entries having covariances
	\begin{align}
		{\rm cov}[T_{ji_1}, T_{ji_2}] &=
		\frac{d_2M^2-1}{d_2^3M^2-d_2} - \frac{1}{d_2^2} < 0 \qquad \forall j \in
		[d_2], \, \forall i_1 \neq i_2 \in [d_1]\\
		{\rm cov}[T_{j_1i_1}, T_{j_2i_2}] &= \frac{M^2}{(d_2M)^2-1} - 
		\frac{1}{d_2^2} > 0 
		\qquad \forall j_1 \neq j_2 \in [d_2], \, \forall i_1 \neq i_2 \in [d_1].
	\end{align}
\end{proposition}

\begin{proof}
	Let us first prove the result on the distribution of the individual columns of $T$. Using the Stinespring representation of the (random) quantum channel $\Phi$ \eqref{eq:def-Stinespring}, we have
	\begin{equation}T_{ji} = \sum_{k=1}^M |V_{M(j-1)+k, i}|^2.\end{equation}
	For a fixed column $i$, the distribution of the elements $T_{ji}$ are
	obtained by summing successive blocks of size $M$ from the vector
	$\ket{v_i} = V_{\cdot i} \in \C^{d_2M}$ which is uniformly distributed on
	the unit sphere of the corresponding vector space. We recognize the partial
	trace operation, and we have
	\begin{equation}T_{ji} = \rho^{(i)}_{jj} := [\operatorname{id}_{d_2}
	\otimes
	\mathrm{Tr}_{M}](\ketbra{v_i}{v_i})_{jj}.\end{equation}
	Using Proposition \ref{prop:diag-random-density-matrix} for the random quantum state $\rho^{(i)}$ having the induced distribution with parameters $(d_2,M)$, we obtain the claim about the individual columns of $T$.
	The covariance expressions can be readily obtained from spherical integration formulas or using the Weingarten calculus, see \cite{cs06} or \cite[Proposition 4.2.3]{hp00}:
	\begin{align}
		\mathbb E |V_{a,i_1}|^2|V_{a,i_2}|^2 &=	\frac{1}{d_2M(d_2M+1)} \qquad \forall a \in [d_2M], \, \forall i_1 \neq i_2 \in [d_1] \\
		\mathbb E |V_{a_1,i_1}|^2|V_{a_2,i_2}|^2 &=\frac{1}{(d_2M)^2-1}  \qquad \forall a_1 \neq a_2 \in [d_2M], \, \forall i_1 \neq i_2 \in [d_1].
	\end{align}
\end{proof}

\begin{remark}\label{rem:channel-stoch-not-flat}
	The columns of the uniform distribution on the set of column-stochastic matrices of 
	size $d_2 \times d_1$ are independent and therefore their entries have covariances
	\begin{equation}
		{\rm cov}[T_{j_1i_1}, T_{j_2i_2}] = 0 \qquad \forall j_1,j_2 \in [d_2], 
		\, \forall 
		i_1 \neq i_2 \in [d_1].
	\end{equation}
	This shows that the distribution $\mu^{Stinespring}_{d_1,d_2;M}$ cannot yield the 
	uniform distribution on the set of (column-)stochastic matrices, for any $M \geq 1$. 
	Note that for $M=1$ and $d_1=d_2=d$, we recover the distribution on the set of 
	\emph{unistochastic matrices}, see \cite{ZKSS03,DZ09,MKZ13}.
\end{remark}

Let us point out that it is possible to generate a random uniform $d_2 \times
d_1$ (column-)stochastic matrix from Gaussian distributions. It is convenient to
start with a  rectangular matrix $G$ of order $d_2 \times d_1$ from the complex
Ginibre ensemble, all elements of which are independent complex random Gaussian
variables. Then a random matrix given by the Hadamard product,  $G\odot {\bar
G}$ contains non-negative  entries $|G_{ji}|^2$. Renormalizing the matrix
according to the sum in each column we introduce a matrix $T$,
\begin{equation}
  T_{ji}=\frac{|G_{ji}|^2}{\sum_{k=1}^{d_2} |G_{ki}|^2},
\end{equation}
which is stochastic by construction. Furthermore, as  each of its columns forms
an independent random vector distributed uniformly in the probability simplex of
size $d_2$, in the case $d_1=d_2=d$ this construction provides random matrix
distributed uniformly in the set of stochastic matrices \cite{BSCSZ10}. Note
however that this is not the procedure described in Section
\ref{sec:different-measures} used to generate random quantum channels from
Gaussian Kraus operators, so the above fact does not contradict Remark
\ref{rem:channel-stoch-not-flat}.

Next, let us briefly discuss the spectral properties of random stochastic
matrices generated according to the  Lebesgue measure, in order to compare them
to the results in Section \ref{sec:spectrum-random-quantum-channels}. Due to the
classical Frobenius-Perron theorem, any stochastic matrix $T$ has the leading
eigenvalue $\lambda_1=1$ which corresponds to the invariant state. Furthermore
it is known that the support of the spectrum of $T$ forms a proper subset of the
unit disk described by the bounds of Karpelevich \cite{kar}, which for a large
dimension covers the  entire unit disk.

However, for a random stochastic matrix of a large dimension $d$, the density of
eigenvalues in the unit disk is not uniform. In order to analyze this issue it
will be useful to distinguish the uniform probability vector, ${\vec
p_*}=(1/d,\dots 1/d)$, and represent any other probability vector ${\vec p}$ in
shifted coordinates, ${\vec p}=  {\vec  p_*} +    {\vec \delta}$. As the
normalization condition implies $\sum_i \delta_i = 0$, the shift vector $\vec
\delta$ has $(d-1)$ independent variables. Let us denote by  $\vec \eta $ the
projection of $\vec \delta$ onto the $(d-1)$ dimensional real space. Hence the
translation vector  ${\vec \eta} \in  {\mathbbm R}^{d-1}$, is analogous to the
generalized Bloch vector $\vec \tau$ used in the quantum case -- see Fig.
\ref{fig:simplex-Dirichlet}b. In  fact, one can treat  ${\vec \eta}$ as the
Bloch representation (\ref{Bloch1}) of a diagonal density matrix,
\begin{equation}
{\rm diag}^\dagger(p) =
\frac1d \left(
\mathbbm 1 + \sum_{j = 1}^{d-1} \eta_j \; \Lambda_j
\right),
\label{Bloch_class}
\end{equation}
where ${\rm \diag}^\dagger(p)$ is the diagonal matrix with $p$ on the diagonal,
the sum goes now over all $(d-1)$ diagonal generators $\Lambda_j$ of
$\textsf{SU(}d\textsf{)}$. In such a Bloch representation of the probability
vector, $p=p({\vec \eta})$, where $\vec \eta$ is formed by the first $d-1$
components of the Bloch vector $\vec \tau$, the action of the classical
transition matrix $T$ can be represented by an affine transformation of the
displacement vector,
\begin{equation}
{\vec \eta'} =   C {\vec \eta} + {\vec \chi},
\label{stoch2}
\end{equation}
analogous to the transformation (\ref{Bloch2}) describing a quantum stochastic
map. Here $C$ represents a real transformation matrix of size $d-1$, and the
translation vector $\vec \chi$ vanishes for any bistochastic matrix. In this
particular basis the stochastic transition matrix reads,
\begin{equation}
T \; =\; \left[
\begin{array}{ll}
  1 & 0 \\
 {\vec  \chi} &{C}
\end{array}
\right] .
\label{stoch3}
\end{equation}
Thus  the spectrum of the stochastic matrix $T$ consists of the Frobenius-Perron
eigenvalue $\lambda_1=1$ and the remaining $d-1$ eigenvalues of a real
asymmetric matrix $C$. For a large dimension,  $d \gg 1$, the constraints
implied on  a matrix $C$ by the fact that $T$ is a random stochastic matrix
become weak, so the statistical properties of $C$ can be approximated by a
matrix $G_\R$ from  the real Ginibre ensemble. Thus the spectrum of a random
stochastic matrix consists of the leading eigenvalue $\lambda_1=1$ and the
scaled disk of complex eigenvalues of Girko \cite{Gi65, Gi84,Fo10}. This simple
reasoning is consistent with the earlier results of Chafa{\"\i} 
\cite{chafai2010dirichlet} and Horvat \cite{Horvat}, who
analyzed an ensemble of random stochastic matrices and found that the spectrum
is concentrated in a disk of radius $r\sim 1/\sqrt{d}$, so the average spectral
gap $\gamma$ behaves as $1-d^{-1/2}$. The circular law for random Markov
operators in the general i.i.d.~case was obtained in \cite{bcc12}.

Note also that the square matrix $C$ of size $d-1$ can be considered as the
classical part of the matrix $Q$  of size $d^2-1$ representing a quantum channel
in Eq. (\ref{super2}), which can be written as
\begin{equation}
Q =
  \left[
\begin{array}{ll}
  C & Q_1 \\
 Q_2 & Q_Q
\end{array}
\right] ,
\label{stoch4}
\end{equation}
Here $Q_1$ and $Q_2$ denote rectangular matrices describing the coupling between
diagonal and off--diagonal parts of the density matrix, while $Q_Q$ is a square
matrix of size $d^2-d$ representing the coupling in the space of coherences. In
this notation expression  (\ref{super2})  for a superoperator  $\Phi$ contains
the transition matrix $T$  as its principal block,
\begin{equation}
{\tilde \Phi} \; =\;
 \left[
\begin{array}{ll}
  1 & 0 \\
{\vec \kappa}&{ Q}
\end{array}
\right]
=
 \left[
\begin{array}{lll}
  1 & 0 &0 \\
 {\vec  \chi} &{C} & Q_1 \\
 {\vec  \chi'} & Q_2  & Q_Q
 \end{array}
\right]
 =
  \left[
\begin{array}{ll}
  T & {\tilde Q_1} \\
 {\tilde  Q_2} &{Q_Q}
\end{array}
\right] ,
\label{stoch5}
\end{equation}
where the Bloch translation vector $\vec \kappa$ of length $d^2-1$ is given by
concatenation of two vectors, $({\vec \chi}, {\vec \chi}')$,  of length $d-1$
and $d^2-d$, respectively, while ${\tilde Q_1}$ and $ {\tilde  Q_2}$ denote the
rectangular matrices $ Q_1,  Q_2$, extended accordingly. Note that in the above
formula the sizes of each block in the first and the last matrix are different.

The above representation implies  that if the core $Q$ of the quantum
superoperator $\Phi$ can be approximated by an uncorrelated  Ginibre matrix
$G_\R$ of size $d^2-1$, its classical block $C$ of order $d-1$ inherits similar
properties. Furthermore,  the {\sl coherence} of a quantum channel $\Phi$ can be
characterized by the $2$--norm of all off-diagonal parts of the corresponding
Choi matrix $J_\Phi$.  In notation used in  Eq. (\ref{stoch5}) the coherence is
given \cite{KCPZ18} by the sum of  squared norms of the `quantum' blocks of the
superoperator, ${\cal C}_2(\Phi)= ||\tilde{Q_1}||_2^2+|| {\tilde Q_2}||_2^2 +
||Q_Q||_2^2$, where $\|\cdot\|_2$ denotes Hilbert-Schmidt norm.

Using Proposition 11 from \cite{NPPZ18}, which bounds the asymptotic difference,
in terms of infinity norm, between random Choi matrix $J_{\Phi}$ and Wishart
matrix, we can conclude, that the asymptotic behavior of coherence measures
for random channels can be derived form the results for the random quantum
states. We consider coherence measures defined for quantum channels  as
\begin{equation}
\begin{split}
{\cal C}_2(\Phi) &= \sum_{i\neq j} |(J_{\Phi})_{ij}|^2, \\
{\cal C}_1(\Phi) &= \sum_{i\neq j} |(J_{\Phi})_{ij}|, \\
{\cal C}_e(\Phi) &= S(\operatorname{diag}(J_{\Phi})) - S(J_{\Phi}),
\end{split}
\end{equation}
where $S$ is an appropriate entropy function.
Then if we take for example random Choi matrix with distribution
$\mu_{d,d}^{Lebesgue}$, following~\cite{puchala2016distinguishability} we have
for $d \to \infty$,
\begin{equation}\label{key}
\begin{split}
{\cal C}_2 (\Phi) &\simeq 1, \\
\frac{1}{d^2} {\cal C}_1 (\Phi) &\simeq \frac{\sqrt{\pi}}{2}, \\
\frac{1}{d} {\cal C}_e (\Phi) &\simeq \frac12.
\end{split}
\end{equation}

\medskip

For any bistochastic quantum map the translation vector $\vec \kappa$ has to
vanish. A method to generate random bistochastic maps, based on a variant of the
Sinkhorn algorithm \cite{Si64}, consisting in alternating rescaling by partial
traces performed on the first and the second subsystem was used in  \cite{AS08}.
However, it remains uncertain, which measure in the set of bistochastic
operations is induced by this procedure, as it is known that in the classical
case  \cite{CSBZ09} the flat measure in the space of bistochastic matrices is
achieved only in the limit of large system size $d$.

Finally, we briefly discuss the properties of the spectrum of  random stochastic
matrices originating from decohered random Choi matrices. Assume we have a
random channel $\Phi$ sampled from $\mu_{d,d}^{Lebesgue}$. We introduce a
decohered channel $\Phi_b$, for some decoherence parameter $b \in [0,1]$ as

\begin{equation}
J_{\Phi_b} =  b J_\Phi + (1-b) \diag(J_\Phi),
\label{beta1}
\end{equation}
so that for $b=1$ one recovers  an original random quantum channel, while the
case of full decoherence with the classical transition matrix $T$ at the
diagonal of the Choi matrix is obtained for $b=0$. The quantum to classical
transition occurs as parameter $b$ decreases from one to zero, but the effective
transition parameter depends also on the dimension $d$.

As shown in Proposition~\ref{prop:T-from-Phi}, the fully decohered uniformly
chosen channel $\Phi$ does not give rise to a uniform distribution on the set of
stochastic matrices, as there exists some correlation between the elements of
the resulting stochastic matrix $T$ and its columns are not uniformly sampled
probability vectors.
In the case of one-qubit random channels the exact distribution
of entries of the stochastic matrix $T$
obtained by superdecoherence of $\Phi$ was derived in \cite{LA17}.

 For a uniformly sampled quantum channel $\Phi$ of an arbitrary dimension $d$
 the distribution of entries of each column
of the corresponding classical stochastic matrix $T=T(\Phi)$ 
can be approximated by the Dirichlet
distribution with parameter $s=d^2$, written
$\operatorname{D_{d^2}}(T_{1j},\ldots, T_{dj})$ on $\Delta_d$ for $j=1,\dots, d$
as visualized for $d=3$ in  Fig. \ref{fig:simplex-Dirichlet}b. The spectra of
random stochastic matrices with columns sampled from the uniform and the
aforementioned distributions are shown in panels \ref{fig:simplex-Dirichlet}c
and~\ref{fig:simplex-Dirichlet}d. The structure of the support of the spectra
visible in panel \ref{fig:simplex-Dirichlet}c, consistent with the equilateral
triangle superimposed with the interval $[-1,1]$ according to the bounds
\cite{kar} for $d=3$, for larger dimensions covers the  entire disk of radius
$r_c\sim d^{-1/2}$.
%%%%
In Fig~\ref{fig:simplex-Dirichlet}e we show, for comparison, the spectra of
uniformly sampled random quantum channels $\Phi$ acting on the states of size
$d=3$.

This distribution implies that the essential spectrum of  such a random
stochastic matrix is highly concentrated around the origin. Numerical
investigations of a stochastic matrix  $T$ obtained by complete decoherence of a
random quantum channel corresponding to Eq.  (\ref{beta1}) with $b=0$, reveal
that the radius of the essential spectrum scales as $d^{-3/2}$ so that the
spectral gap $\gamma$ behaves as $1 - d^{-3/2}$. Fig.~\ref{fig:beta-dep} shows
the numerically found behavior of the radius of the essential spectrum of the
form $d^{\alpha(b)}$ as a function of the parameter $b$. Note the rapid
transition between the regime $\alpha=-3/2$ and $\alpha=-1$, which we will
discuss in the next paragraph.

\begin{figure}[!h]
	\centering\includegraphics{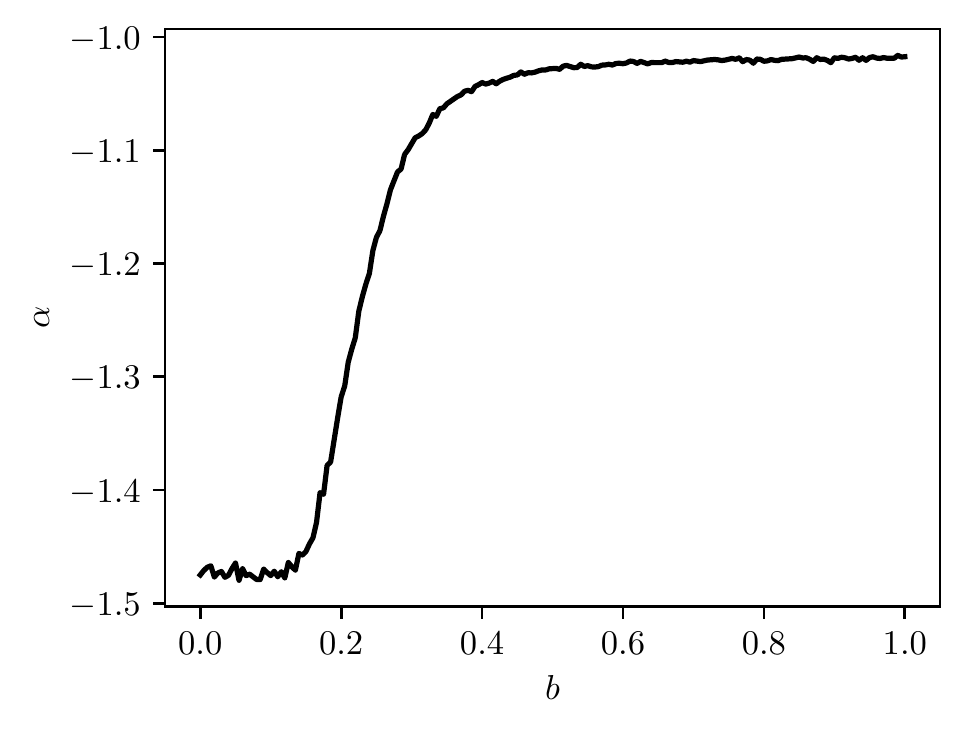} \caption{Numerically found dependence
	of the exponent $\alpha$ for the model $r\propto d^{\alpha(b)}$. The plot
	was obtained by fitting the model $r \propto d^{\alpha(b)}$ for $d \in [10,
	40]$, $b \in [0, 1]$ with step $0.005$ and 100 samples for each $d$ and
	$b$}\label{fig:beta-dep}
\end{figure}

Let us first analyze this result in the case of complete decoherence,
 corresponding to  $b=0$. Consider a random variable $x$, a component of a
 random probability vector of size $d$ distributed according to the Dirichlet
 distribution of order $s$. Consider the variable $y = x - 1/d$, satisfying $\E
 y=0$, for which $\E y^2 = \operatorname{Var}(x) \sim 1/(s d^2)$. Now, we
 construct a Ginibre-like matrix $\tilde{G}$ of size $d \times d$ filled with
 independent, identically distributed variables $y$, which mimics the classical
 stochastic transition matrix $T$. Observe that
\begin{equation}
	\E \|\tilde{G}\|_2^2 = d^2 \E y^2 \sim \frac{d^2}{sd^2} = \frac{1}{s}.
\end{equation}
Hence the radius of the Girko disk describing the spectrum of $\tilde{G}$ of
size $d$ reads
\begin{equation}
	r_s = \sqrt {\frac{\|\tilde{G}\|_2^2}{d}} = \sqrt{\frac{1}{sd}}.
\end{equation}
For a large  dimension $d$ we can assume that a random Ginibre matrix $\tilde G$
describes spectral properties of the core $C$ of the  random transition matrix
(\ref{stoch3}) generated according to the Dirichlet distribution $D_s$. In the
case $s=1$ (flat distribution) we have $r_1 \sim d^{-1/2}$ in accordance to the
results of Chafa{\"\i} \cite{chafai2010dirichlet}, Horvat~\cite{Horvat}, and 
Bordenave \etal~\cite{bcc12}. The case 
$s=d$
yields $r_s\sim d^{-1}$. For the completely decohered case we have $s=d^2$ and
we recover the numerically found result $r_{d^2} \sim d^{-3/2}$. Additionally, a
more general case $s=d^k$ yields $r_{d^k}\sim d^{-(k+1)/2}$. During the
transition from random quantum channels $\Phi$ to classical stochastic
transition matrices $T$ the radius of the essential spectrum scales with
dimension as $r\sim d^{\alpha}$,  with the exponent decreasing from $-1$
(quantum) to $-3/2$ (classical).  This transition is is visualized  in
Fig.~\ref{fig:beta-dep}, which shows how exponent $\alpha$ changes with the
decoherence parameter $b$.

In order to estimate the effective scaling parameter one can expect that the
transition ocurres as the size $r_q$ of the ``quantum disk''  related to all
$d^4-d^2$ entries of the superoperator related to off-diagonal elements of the
Choi matrix is comparable with the radius $r_c$ of the `classical  disk'
associated with the the stochastic matrix $T$ with $d^2$ elements obtained by
reshaping the diagonal of $J$.  As the norm of the off-diagonal  part behaves as
$ \|(J_\Phi)_\mathrm{off} \|_2^2 = \sum_{i \neq j} |(J_\Phi)_{ij}|^2 \sim b^2$,
we obtain an estimate  for  the `quantum' Girko disk, $r_q =
\sqrt{\|J_\mathrm{off}\|_2^2/d^2} =b/d$. As discussed above the radius of the
`classical' Girko disk behaves as $r_c = d^{-3/2}$. Setting $r_c=r_q$ we arrive
at a relation for the critical value of the decoherence parameter, $b_c  \approx
1/\sqrt{d}$, which implies behavior of the dimension-independent scaling
parameter $\beta= b \sqrt{d}$.

To drive this point further we present numerically obtained fraction, $p$ of the
bulk eigenvalues of a uniformly chosen $\Phi$, contained in the ``classical''
disc, $|\lambda| \leq 1/d^{-3/2}$, as a function of the super-decoherence
parameter $b$ for a few values of the dimension $d$, see Fig.~\ref{fig:4}a. In
consistence with the results discussed before, for small values of $b$ this
fraction stays constant and than decreases sharply for a critical value $b_c$,
which depends on $d$. Fig.~\ref{fig:4}b shows the same results as a function of
the proposed scaling parameter, $\beta=b\sqrt{d}$. This reveals that the scaling
of the critical value $b_c \approx 1/\sqrt{d}$ is correct, as all data merge
into a single curve. These numerical results were obtained using the
\verb|QuantumInformation.jl| package~\cite{GKP18}.
Analogous scaling of the critical value of the transition parameter
with the dimension was reported recently
\cite{TYLDCZ21} while studying 
superdecoherence of random Lindblad operators.

 \begin{figure}[!htp]
	\centering\includegraphics[width=0.48\textwidth]{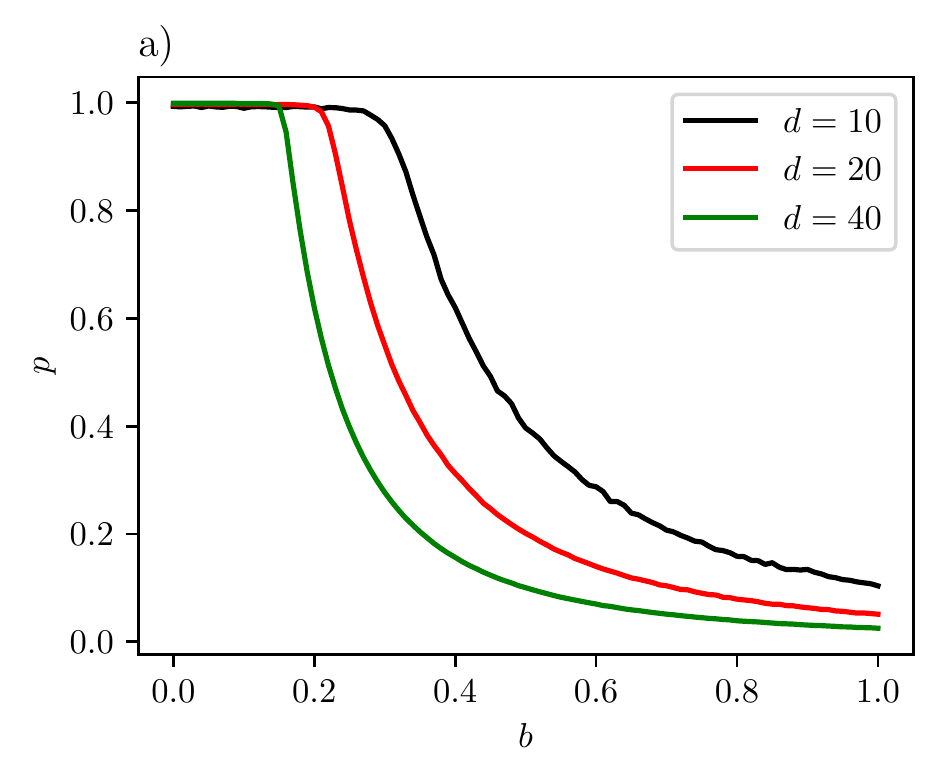}
	\includegraphics[width=0.48\textwidth]{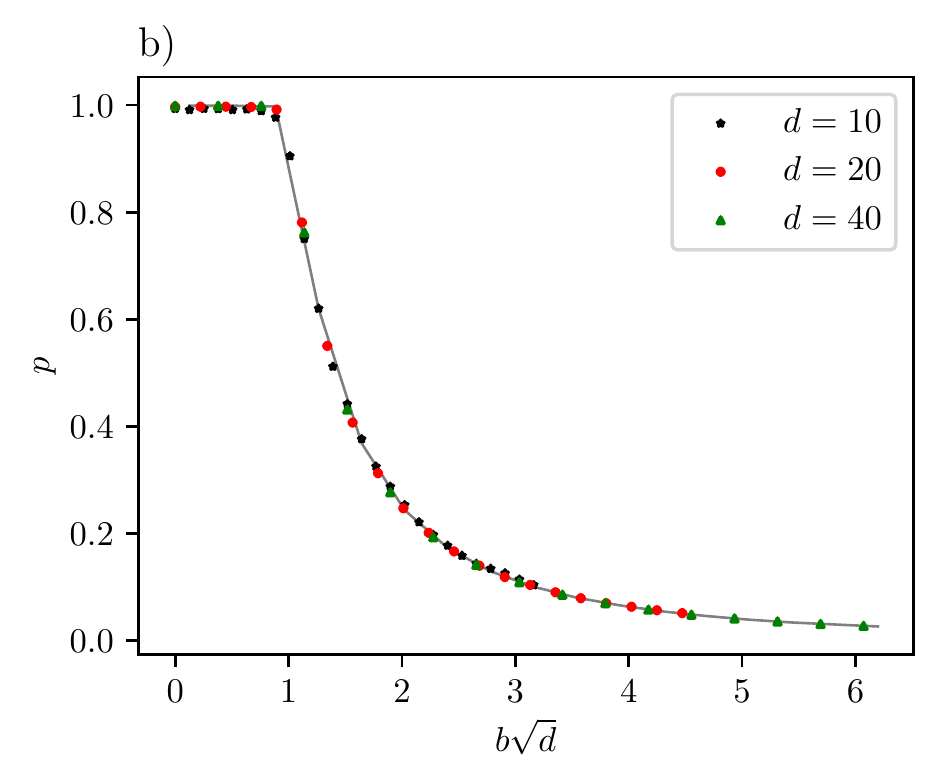}
	\caption{Plots of the fraction of eigenvalues $p$ of a randomly chosen channel
	$\Phi$ contained in the `classical'  disc $|\lambda| \leq 1/d^{-3/2}$.
	Panel ~a) presents results as a function of the super-decoherence parameter $b$,
	whereas panel~b) presents them as the function of the parameter
	 $\beta=b\sqrt{d}$ and confirms the scaling, $b_c\approx 1/\sqrt{d}$.
	 For clarity only a handful of the data is plotted and a curve is added
	   to guide the eye.}\label{fig:4}
 \end{figure}

\section{Invariant states of random quantum channels}
\label{sec:inv}

In this section we will focus on the properties of the invariant state of random
quantum maps sampled according to the measures discussed in Section
\ref{sec:different-measures}. Properties of the superoperator associated with
these maps was discussed in Section~\ref{sec:spectrum-random-quantum-channels}.
Here we will provide an in-depth analysis of the Perron-Frobenius invariant
state of these maps.

One of the properties characterizing random quantum maps is their ability of
mixing quantum states. The image of the maximally mixed state
$\mathbbm{1}_{d_1}/{d_1}$ under random quantum maps will be concentrated around
the maximally mixed state $\mathbbm{1}_{d_2}/{d_2}$, so a random stochastic
channel is close to be bistochastic (see Lemma \ref{lem-1} for a quantitative
statement). Thus, one can expect that the invariant state $\rho_{\rm{inv}}$ of a
random map $\Phi$ can be found in the neighborhood of the maximally mixed state
(see Fig.~\ref{fig:3} for the numerical evidence). We formalize this statement
in the theorem below.

\begin{theorem}\label{th:inv}
	Let $\Phi$ be a random channel sampled according to $\mu_{d,d}^{Lebesgue}$
	with the unique invariant state $\rho_{\rm{inv}}$. As $d \to \infty$,
	the invariant state
	convergences almost surely in the trace norm to the maximally mixed state,
	\begin{equation}
	\Vert \rho_{\rm{inv}}-\mathbbm{1}_d/d \Vert_1 = O(1/d).
	\end{equation}
\end{theorem}

\begin{figure}[!h]%disk
  \includegraphics[angle=0,width=0.5\columnwidth]{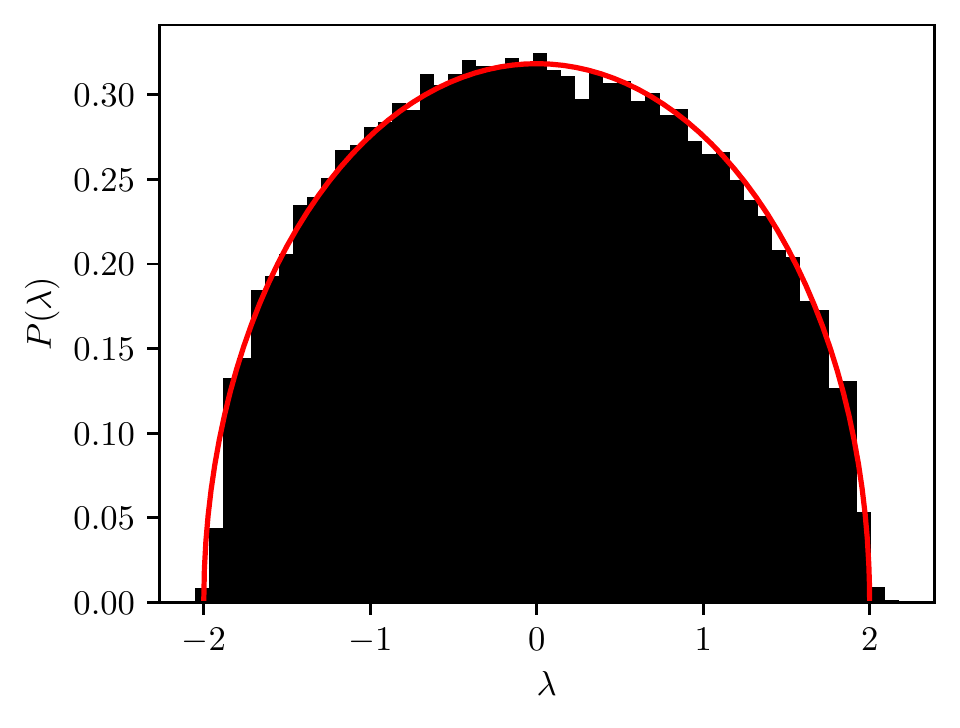}
  \caption{Illustration of Theorem~\ref{th:inv}. We show the level density of
   the matrix $d^2(\rho_\mathrm{inv} - \1_d/d)$ where $\rho_{\mathrm{inv}}$ is
   the unique invariant state of a random channel $\Phi$ sampled according to
   the flat distribution $\mu_{d,d}^{Lebesgue}$. The histogram is a Monte Carlo
   simulation obtained for $d=40$ with 1000 independent random channels. Red
   curve represents the Wigner semicircular law.}
  \label{fig:3}
\end{figure}

Before we prove Theorem \ref{th:inv}, let us first establish a few lemmas, to
which we will refer in the main proof. We first prove that random quantum
operations are contraction maps. We examine behavior of the Lipschitz constant
(with respect to the Schatten 1-norm) which is defined as minimum over constants
$L$ satisfying $\Vert \Phi (\rho-\sigma)\Vert_1 \leq L \Vert \rho-\sigma
\Vert_1$ for all states $\rho, \sigma$. In order to prove this lemma we need to
introduce the \emph{diamond norm} of a linear map $\Phi: M_{d_1}(\C) \to
M_{d_2}(\C)$:
\begin{equation*}
	\| \Phi \|_\diamond =
	\sup_{X \neq 0} \frac{\| \left(\Phi \otimes \operatorname{id}_{d_1} \right)(X)\|_1}{\| X \|_1}.
\end{equation*}
For a Hermiticity-preserving map $\Phi$, it suffices to optimize over pure
states, see \cite[Chapter 3.3]{Wat18}. We have the following upper bound on the
Schatten 1-norm Lipschitz constant.

\begin{lemma}\label{lem:Lipschitz}
	Let $L_\Phi$ denote the Lipschitz constant of random quantum operation
	$\Phi$ sampled according to $\mu_{d,d}^{Lebesgue}$. Then, almost surely, as
	$d\to \infty$
	\begin{equation}
	L_\Phi \leq \frac{3
		\sqrt{3}}{2 \pi}<1.
	\end{equation}
\end{lemma}
\begin{proof}
	One can obtain
	\begin{equation}
	\begin{split}
	L_{\Phi} &= \max_{\rho \not = \sigma} \left\| \Phi \left(
	\frac{\rho-\sigma}{\Vert \rho - \sigma \Vert_1}\right)
	\right\|_1=\max_{\substack{H=H^\dagger \\ \tr H=0 \\ \Vert H \Vert_1=1}}
	\Vert
	\Phi (H)
	\Vert_1 \\
	&= \frac12 \max_{\ket{x} \bot \ket{y}} \Vert \Phi (\proj{x}-\proj{y})
	\Vert_1.
	\end{split}
	\end{equation}
	In the next step we will use the diamond norm,
	$\Vert  \cdot 	\Vert_\diamond$,
	to bound this value:
	\begin{equation}
	\begin{split}
	L_\Phi & \leq \max_{\ket{x}} \Vert \Phi (\proj{x}) - \1_d/ d \Vert_1 =
	\max_{\ket{x}} \Vert (\Phi-\Phi_*) (\proj{x}) \Vert_1 \leq \Vert
	\Phi-\Phi_*	\Vert_\diamond, \end{split}
	\end{equation}
	where $\Phi_*$ denotes maximally depolarizing channel. By \cite[Theorem
	16]{NPPZ18} we have $\Vert \Phi-\Phi_* \Vert_\diamond
	\xrightarrow{}
	\frac{3
		\sqrt{3}}{2 \pi}<1$, proving the claim.
\end{proof}

The next lemma gives an upper bound on the distance between the maximally mixed
state and its image through a random quantum channel. Note that it is stated in
a slightly more general setting.
\begin{lemma}\label{lem:dist-max-mixed}
	Let $\Phi$ be a random channel sampled according to
	$\mu_{d,d;M_d}^{Stinespring}$, where $M_d \geq 1$ is a sequence of integers satisfying $M_d \sim td^2$ as $d \to \infty$, for some constant $t > 0$. Then, we have
	\begin{align*}
	\E \Tr \left( \left( \Phi(\1_d/d) - \1_d/d \right)^2 \right) &=
	\frac{(d^2-1)(td^2-1)}{d(t^2d^6-1)} \sim \frac{1}{td^3}\\
	\operatorname{Var}\left( (td^3) \Tr \left( \left( \Phi(\1_d/d) - \1_d/d
	\right)^2 \right) \right) &= \frac{2}{d^2} + O(d^{-4}).
	\end{align*}
	In particular, almost surely,
	$$\lim_{d \to \infty} (td^3)\|\Phi(\1_d/d) - \1_d/d\|_2^2 = 1.$$
\end{lemma}
\begin{proof}
	The first claim follows from Lemma \ref{lem:Wg} with the choice
	$A=B=\1_d/d$, see also Proposition~\ref{prop:purity-unitarity}. The second
	claim is proven using the Weingarten calculus to compute fourth moments of
	$\Phi(\1_d/d)$; we provide in the Supplementary Material a
	\texttt{Mathematica} notebook which performs this tedious computation using
	the \href{https://github.com/MotohisaFukuda/RTNI}{RTNI} package provided in
	\cite{fkn19}. Finally, the almost sure convergence follows from the
	Borel-Cantelli lemma.
\end{proof}

Now we are ready to prove Theorem \ref{th:inv}.
\begin{proof}[Proof of Theorem \ref{th:inv}]
	Using Lemmas \ref{lem:Lipschitz} and \ref{lem:dist-max-mixed}, we have that, almost surely as $d \to \infty$,
	\begin{equation}
	\begin{split}
	\Vert \rho_{\rm{inv}} - \mathbbm{1}_d/d \Vert_1 &= \Vert
	\Phi^k(\rho_{\rm{inv}}) -
	\Phi^k(\mathbbm{1}_d/d) +
	\Phi^k(\mathbbm{1}_d/d) -
	\Phi^{k-1}(\mathbbm{1}_d/d) + \ldots - \Phi(\mathbbm{1}_d/d) +
	\Phi(\mathbbm{1}_d/d) - \mathbbm{1}_d/d \Vert_1 \\
	& \leq 2 L_\Phi^k+\frac{1-L_\Phi^k}{1-L_\Phi} \Vert
	\Phi(\mathbbm{1}_d/d) - \mathbbm{1}_d/d
	\Vert_1
	\xrightarrow{k \to \infty} \frac{\Vert \Phi(\mathbbm{1}_d/d) -
		\mathbbm{1}_d/d
		\Vert_1}{1-L_\Phi}\\
	&\leq \frac{d^{1/2}\|\Phi(\1_d/d)-\1_d/d\|_2}{1-L_\Phi} = O(d^{-1}).
	\end{split}
	\end{equation}
\end{proof}

We can also show the following stronger result, which we can only prove in the
asymptotical regime $1 \ll d_1 \ll d_2$.

\begin{lemma}\label{lem-1}
	Let $\Phi$ be a random channel sampled according to
	$\mu_{d_1,d_2;td_1d_2}^{Choi}$, where $t > 0$ is some constant value. In the regime
	$1 \ll d_1 \ll d_2$, the Hermitian non-unitality shift  matrix, $\delta
	H:=\sqrt{t}d_1d_2(\Phi(\mathbbm{1}_{d_1}/d_1)-\mathbbm{1}_{d_2}/d_2)$
	converges in moments toward the standard semicircular distribution.
\end{lemma}
\begin{proof}
	We can write
	\begin{equation}
	\sqrt{t}d_1d_2(\Phi(\mathbbm{1}_{d_1}/d_1)-\mathbbm{1}_{d_2}/d_2)=\sqrt{t}
	d_2\left([\operatorname{id}_{d_2} \otimes \Tr_{d_1}]\left(J_\Phi -
	\frac{1}{t d_1
		d_2^2}W\right)\right)+\sqrt{t}d_1\left(\frac{1}{t d_1^2 d_2}
	[\operatorname{id}_{d_2} \otimes \Tr_{d_1}] W -
	\mathbbm{1}_{d_2}\right),
	\end{equation}
	where we introduced a Wishart matrix $W$ with parameters $(d_1d_2,td_1d_2)$
	such that the Choi matrix $J_\Phi$ is obtained by a partial normalization of $W$
	like in the Definition \ref{def:random-Choi}. By the Corollary $2.5$ in
	\cite{cnye} we know that the matrix $\sqrt{t}d_1\left(\frac{1}{t d_1^2 d_2}
	[\operatorname{id}_{d_2} \otimes \Tr_{d_1}] W - \mathbbm{1}_{d_2}\right)$
	converges to the standard semicircular
	distribution. To finish the proof we need to bound the first term of the sum:
	\begin{equation}\label{eq:inv}
	\begin{split}
	\sqrt{t} d_2 \left\Vert [\operatorname{id}_{d_2} \otimes
	\Tr_{d_1}]\left(J_\Phi - \frac{1}{t d_1
		d_2^2}W\right)\right\Vert_\infty&=\sqrt{t}
	d_2 \left\Vert [\operatorname{id}_{d_2} \otimes
	\Tr_{d_1}]\left(W\left(\mathbbm{1}_{d_2} \otimes
	\left(([\Tr_{d_2} \otimes \operatorname{id}_{d_1}]W)^{-1} -
	\frac{1}{t d_1
		d_2^2}\mathbbm{1}_{d_1}\right)\right)\right)\right\Vert_\infty\\
	&\leq \sqrt{t}d_2 \Vert [\operatorname{id}_{d_2} \otimes \Tr_{d_1}]W
	\Vert_\infty \left\Vert ([\Tr_{d_2} \otimes \operatorname{id}_{d_1}]W)^{-1} -
	\frac{1}{t d_1
		d_2^2}\mathbbm{1}_{d_1}  \right\Vert_\infty.
	\end{split}
	\end{equation}

	One can note that $[\operatorname{id}_{d_2} \otimes \Tr_{d_1}]W$ is a
	Wishart matrix with parameters
	$(d_2,td_1^2d_2)$ and $[\Tr_{d_2} \otimes \operatorname{id}_{d_1}]W$ is a
	Wishart matrix with parameters
	$(d_1,td_1d_2^2)$. According to Theorem 2.7 in \cite{cnye} we have that in
	the limit $d_1,d_2 \to \infty$, the spectrum of the matrix
	$[\Tr_{d_2} \otimes \operatorname{id}_{d_1}]W/(td_1d_2^2)$ concentrates
	around value one with the convergence
	rate $\Vert [\Tr_{d_2} \otimes \operatorname{id}_{d_1}]W/(td_1d_2^2) -
	\mathbbm{1}_{d_1} \Vert_\infty =
	O(1/d_2)$, therefore
	\begin{equation}
	\left\Vert ([\Tr_{d_2} \otimes \operatorname{id}_{d_1}]W)^{-1} -
	\frac{1}{t d_1
		d_2^2}\mathbbm{1}_{d_1}  \right\Vert_\infty = O(1/(d_1d_2^3)).
	\end{equation}
	On the other hand $\Vert [\operatorname{id}_{d_2} \otimes \Tr_{d_1}]W
	\Vert_\infty = \Omega(d_1^2d_2)$ (see
	Theorem 2.7 in \cite{cnye}), hence
	we can upper bound the norm in Eq. \eqref{eq:inv} by $O(d_1/d_2)$.
\end{proof}

Numerical results suggest that a stronger convergence with the rate $O(1/d_2)$
holds in Eq.~\eqref{eq:inv}; such a result would imply that the conclusion of
the Lemma holds under the more natural assumption $d_1, d_2 \to \infty$. We
leave this question open.

\section{Concluding remarks}

In this work we analyzed various techniques of generating random quantum
channels. On the one hand we have investigated three natural  methods and showed
under which choice of parameters they become equivalent and induce the desired
flat Lebesgue measure in the set $\CC_{d,d}$ of all quantum operations acting of
density matrices of order $d$. On the other hand, we revealed previously
unexplored properties of random channels, analyzing the invariant state and
showing that due to the measure concentration phenomenon it converges
asymptotically to the maximally mixed state. Furthermore, we estimated typical
deviations of a random stochastic channel from unitality and showed that, almost
surely, in the limit of large input and output dimension, and for a particular
scaling of the size of the environment, it becomes unital.

The spectrum of superoperator $\Phi$ representing a random quantum channel of
size $d$ consist of the leading Frobenius-Perron eigenvalue $\lambda_1=1$ and
the remaining $d^2-1$ complex eigenvalues, which typically belong to the Girko
circular disk of radius $r_q\approx d^{-1}$. In the case, $d \gg 1$, the
distribution of these eigenvalues becomes uniform, which is related to the fact
that the core $Q$ of the superoperator matrix can be asymptotically approximated
by a real random Ginibre matrix of size $d^2-1$.

Any quantum channel $\Phi$ determines a classical stochastic transition matrix
$T$ encoded in the diagonal of the corresponding Choi matrix $J_{\Phi}$. The
matrix $T$ of order $d$ can also be considered as a truncation of the
superoperator $\Phi$ of size $d^2$. Hence a given ensemble of random quantum
channels induces an ensemble of random classical transition matrices. We
investigated related ensembles of random quantum channels and the associated
classical counterpart of random stochastic matrices and demonstrated common
properties of the spectra of operators used in both set-ups. Representing a
$d$-point probability distribution  by its displacement vector $\vec \eta$ of
length $d-1$, analogous to the generalized Bloch vector $\vec \tau$, the action
of the classical transition matrix $T$ is then represented by a real matrix $C$
of size $d-1$. As this matrix can be obtained as a truncation of the matrix $Q$
representing the quantum stochastic map $\Phi$ and mimicked by a random Ginibre
matrix, the spectral properties  of $T$ are analogous and the bulk of the
spectrum conforms to the circular law. Note that the radius of the Girko disk
for random stochastic matrices of size $d$ distributed uniformly scales as $r_1
\approx d^{-1/2}$, while it behaves like  $r_c \approx d^{-3/2}$ for transition
matrices obtained due to decoherence of random quantum channels.

We analyzed also some quantities characterizing typical quantum channels. In
particular, we derived expressions for the average output purity,  the average
unitarity (\ref{unitarit}) and the average $2$--norm coherence of a random
quantum channel acting on a $d$ dimensional state.

\medskip

We would like to conclude the paper presenting a list of open questions,
especially regarding the connection between random quantum channels and random
stochastic matrices. These include
\begin{enumerate}
\item Which measure in the space of bistochastic operations induces the flat
 measure in space of bistochastic matrices?

\item Determine whether  for a large matrix size $d$ a random bistochastic
operation generated  by  the Sinkhorn algorithm  \cite{Si64} applied \cite{AS08}
to an initial random stochastic channel $\Phi$, covers uniformly  the entire set
of bistochastic operations, and whether the corresponding random transition
matrix $B$ obtained by a super-decoherence of a quantum channel, covers
uniformly the Birkhoff polytope ${\cal B}_d$ of bistochastic matrices?

\item Is there a purely classical ensemble reproducing the distribution of
(column-)stochastic matrices from Proposition~\ref{prop:T-from-Phi}?

\item Is a random state $\sigma$ of size $d^2$ asymptotically $*$-free
\cite{VDN92} from its reshuffling $\sigma^R$ in the limit $d\to \infty$ ? Such a
conjecture, analogous to the result of  Mingo and Popa  \cite{MP16} concerning a
Haar random unitary matrix and its transpose  (see also \cite{MPS20}), is
supported by numerical results presented in \cite{MLZ18}.
\end{enumerate}

\bigskip
\section*{Supplementary Material}

See supplementary material for provided the Weingarten calculus computed by 
using
	the \texttt{Mathematica} and \href{https://github.com/MotohisaFukuda/RTNI}{RTNI} 
	package. 

\section*{Acknowledgments}

It is a pleasure to thank Sergey Denisov and Dariusz Chru{\'s}ci{\'n}ski for
numerous discussions on decoherence of quantum dynamics and the transition
between quantum channels and classical stochastic matrices. We would like to 
thank anonymous referee for in depth analysis and valuable remarks. This 
research was supported by National Science Center in Poland under the Maestro 
grant number DEC-2015/18/A/ST2/00274, SONATA BIS grant number 
2016/22/E/ST6/00062 and by the
Foundation for Polish Science under the grant Team-Net NTQC number 
POIR.04.04.00-00-17C1/18-00.

\section*{Data Availability}
Data sharing is not applicable to this article as no new data were created or analyzed in 
this study.

\appendix

\section{Proof of Proposition \ref{prop:Choi-vs-Kraus}} \label{appendix-1}
\begin{proof}
	Let $\{A_i=G_iH^{-1/2}\}_{i=1}^M$ be the set of the random Kraus operators
	defined as in Definition \ref{def:random-Kraus}. The corresponding channel
	$\Phi$ has distribution $\mu_{d_1,d_2;M}^{Choi}$ and the Choi matrix
	$J_\Phi$
	can be expressed in the terms of given Kraus operators as
	\begin{equation}
	\begin{split}
	J_\Phi =&\sum_{i=1}^M |A_i \rangle\rangle \langle\langle A_i|=\sum_{i=1}^M
	(\mathbbm{1}_{d_2} \otimes (H^\top)^{-1/2}) |G_i \rangle\rangle
	\langle\langle
	G_i| (\mathbbm{1}_{d_2} \otimes (H^\top)^{-1/2})=\\
	&(\mathbbm{1}_{d_2} \otimes (H^\top)^{-1/2}) W
	(\mathbbm{1}_{d_2} \otimes (H^\top)^{-1/2}),
	\end{split}
	\end{equation}
	where we have used the vectorized form $|A_i\rangle \rangle$ of a matrix
	$A_i$
      and introduced the matrix
	\begin{equation}W := \sum_{i=1}^M |G_i \rangle\rangle \langle\langle
	G_i|.\end{equation}
	Since $|G_i \rangle\rangle$ are i.i.d.~complex Gaussian vectors, thus the 
	matrix 
	$W$ has a Wishart
	distribution of parameters $(d_1d_2, M)$. Moreover, we have
	\begin{equation}H^\top = \sum_{i=1}^M G_i^\top \bar{G_i} =[\Tr_{d_2}
	\otimes
	\operatorname{id}_{d_1}] W,\end{equation}
	which proves the claim.
\end{proof}
\section{Proof of Proposition \ref{prop:Stinesping-equals-Kraus}}
\label{appendix-2}

\begin{proof}
	Let $M$ be an integer such that $M \geq d_1/d_2$.
	Let us take a random Haar isometry $V:
	\mathbb{C}^{d_1}
	\to \mathbb{C}^{d_2} \otimes \mathbb{C}^{M}$ as introduced in Definition
	\ref{def:random-Stinespring}. We will use the fact that for
	a random complex Ginibre matrix $G$
	of order  $d_2M \times d_1$
	 the matrix $G(G^\dagger G)^{-1/2}$ is a
	random Haar isometry. That means, we consider $V=G(G^\dagger G)^{-1/2}$ and
	the Kraus decomposition of the channel $\Phi$ defined by $V$ is determined
	by
	operators $A_i=
	\mathbbm{1}_{d_2} \otimes \bra{i} V$ for $i=1,\ldots,M$. Once more, we can
	calculate the Choi matrix $J_\Phi$
	\begin{equation}
	\begin{split}
	J_\Phi &=\sum_{i=1}^M |\mathbbm{1}_{d_2} \otimes \bra{i} V \rangle \rangle
	\langle \langle
	\mathbbm{1}_{d_2} \otimes \bra{i} V |=(\mathbbm{1}_{d_2} \otimes
	V^\top)\sum_{i=1}^M
	|\mathbbm{1}_{d_2} \otimes \bra{i} \rangle \rangle \langle \langle
	\mathbbm{1}_{d_2} \otimes \bra{i} |
	(\mathbbm{1}_{d_2} \otimes \bar{V})=\\
	&=(\mathbbm{1}_{d_2} \otimes
	(G^\top \bar{G})^{-1/2})(\mathbbm{1}_{d_2} \otimes
	G^\top)(|\mathbbm{1}_{d_2} \rangle\rangle \langle\langle \mathbbm{1}_{d_2}
	|
	\otimes \mathbbm{1}_M) (\mathbbm{1}_{d_2} \otimes
	\bar{G})(\mathbbm{1}_{d_2} \otimes
	(G^\top \bar{G})^{-1/2}).
	\end{split}
	\end{equation}
	Denote by $\widetilde{G}=(\mathbbm{1}_{d_2} \otimes
	G^\top) (|\mathbbm{1}_{d_2}
	\rangle\rangle \otimes \mathbbm{1}_M).$ This matrix turns out to be
	a random Ginibre matrix of size $d_1d_2
	\times M$  due to
	\begin{equation}
	\bra{i} \otimes \bra{j} \widetilde{G} \ket{k}= (\bra{i} \otimes
	\bra{j}G^\top)(|\mathbbm{1}_{d_2} \rangle \rangle \otimes \ket{k})=
	\bra{i} \otimes \bra{k} G\ket{j}.
	\end{equation}
	Moreover, the following relation holds,
	\begin{equation}
	\begin{split}
	[\Tr_{d_2} \otimes
	\operatorname{id}_{d_1}](\widetilde{G}\widetilde{G}^\dagger)=&\sum_{i=1}^{d_2}
	 (\bra{i}
	\otimes
	G^\top)(|\mathbbm{1}_{d_2}\rangle\rangle\langle\langle\mathbbm{1}_{d_2}|
	\otimes
	\mathbbm{1}_M)(\ket{i} \otimes \bar{G})=G^\top \bar{G},
	\end{split}
	\end{equation}
	which implies the claim.
\end{proof}
\section{Proof of Proposition \ref{prop:Lebesgue-vs-others}}
\label{appendix-3}

\begin{proof}
	The last equality which we will obtain is $\mu_{d_1,d_2}^{Lebesgue}=
	\mu_{d_1,d_2;d_1d_2}^{Choi}$.
	We use standard calculus methods to obtain the distribution of $J_\Phi$
	(see
	\cite{BCSZ09}). Let $f_{J_\Phi}(D)$ be the probability density function of
	the random Choi	matrix $J_\Phi$ generated according to  Definition
	\ref{def:random-Choi} at the point $D$,
	\begin{equation}
	\begin{split}
	f_{J_\Phi}(D) \propto& \int \delta(J_\Phi-D) \exp(-\Tr GG^\dagger) dG\\
	\propto& \int \int \delta(\mathbbm{1}_{d_2} \otimes H^{-1/2} GG^\dagger
	\mathbbm{1}_{d_2} \otimes H^{-1/2}-D) \delta(H-[\Tr_{d_2} \otimes
	\operatorname{id}_{d_1}](GG^\dagger))
	\exp(-\Tr H)
	dH dG\\
	\propto& \int \int \delta(\sqrt{D} GG^\dagger \sqrt{D}-D)
	\delta(H-\sqrt{H}[\Tr_{d_2} \otimes
	\operatorname{id}_{d_1}](\sqrt{D}GG^\dagger\sqrt{D}) \sqrt{H}) \exp(-\Tr H)
	\det H^{d_2M} \det D^M dH dG \\
	\propto& \int \int \delta(GG^\dagger - \mathbbm{1}_{d_2d_1})
	\delta(\mathbbm{1}_{d_1}-[\Tr_{d_2} \otimes
	\operatorname{id}_{d_1}]D) \exp(-\Tr H) \det H^{d_2 M-d_1} \det
	D^{M-d_1d_2}
	dH dG \\
	\propto& \delta(\mathbbm{1}_{d_1}-[\Tr_{d_2} \otimes
	\operatorname{id}_{d_1}]D) \det D^{M-d_1d_2}.
	\end{split}
	\end{equation}

	In the particular case $M=d_1d_2$
	the exponent vanishes and we arrive at the desired result.
\end{proof}

\section{Average trace of products of outputs of random quantum
channels}\label{app:lemma-Wg}

We state and prove in this Appendix a lemma which might be of an independent
interest.

\begin{lemma}\label{lem:Wg}
	Consider integers $d_1, d_2, M$ such that $M \geq d_1/d_2$, and let $\Phi :
	M_{d_1}(\C) \to M_{d_2}(\C)$ be a random quantum channel having
	distribution $\mu^{Stinespring}_{d_1,d_2;M}$. Then, for any matrices $A,B
	\in M_{d_1}(\C)$, we have:
	\begin{equation}\label{eq:E-Phi-A-B}
	\E \Tr\left(\Phi(A)\Phi(B)\right) = \frac{(\Tr A)(\Tr B)d_2(M^2-1) +
	\Tr(AB)M(d_2^2-1)}{(d_2M)^2-1}.
	\end{equation}
\end{lemma}
\begin{proof}
	The proof is a standard application of the graphical Weingarten calculus
	introduced in \cite{cn10}. The left hand side of Eq.~\eqref{eq:E-Phi-A-B} is
	represented in diagrammatic notation in Figure \ref{fig:E-Phi-A-B}.

	\begin{figure}[th]
		\centering
		\includegraphics[width=0.7\linewidth]{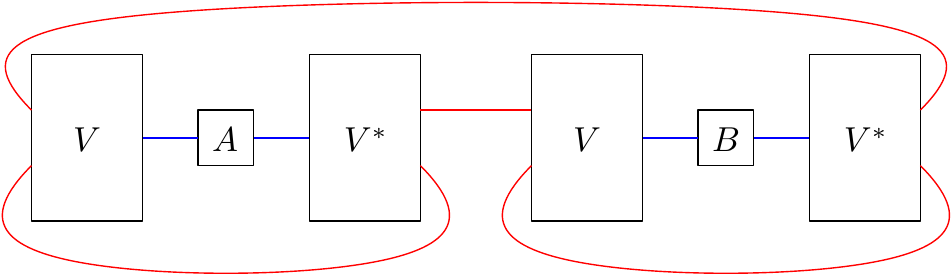}
		\caption{The diagram corresponding to $\Tr[\Phi(A)\Phi(B)]$. When
		computing the expectation of this diagram with respect to the
		Haar-distributed random isometry $V$,
		(visualized as a block with two inputs and one output),
		we use the permutation $\textcolor{red}{\alpha}$ to pair the outer
		(red) wires, and the permutation $\textcolor{blue}{\beta}$ to pair the
		inner (blue) wires.
		The Hilbert space dimensions are as follows: outer, upper wires: $d_2$;
		outer, lower wires: $M$; inner wires: $d_1$.}
		\label{fig:E-Phi-A-B}
	\end{figure}

According to \cite[Theorem 4.1]{cn10}, the average trace in the statement can be
decomposed as a weighted sum of diagrams
$$\E \Tr\left(\Phi(A)\Phi(B)\right) = \sum_{\alpha,\beta \in \mathcal S_2}
\mathcal
D_{\alpha, \beta} \mathrm{Wg}_{d_2M}(\alpha^{-1}\beta),$$
where the Weingarten function on $\mathcal S_2$ reads
\begin{equation}\mathrm{Wg}_{N}((1)(2)) = \frac{1}{N^2-1} \qquad \text{and}
\qquad \mathrm{Wg}_{N}((12)) = \frac{-1}{N(N^2-1)}\end{equation}
and the four diagram are obtained as follows. The permutation $\beta$ is used
to pair the inner wires connected to the $A$ and $B$ boxes. The contributions
of these diagrams are, for the identity permutation and for the transposition
respectively,
\begin{equation}\beta = (1)(2): \, (\Tr A)(\Tr B) \qquad \text{and} \qquad
\beta = (12): \, \Tr(AB).\end{equation}
Similarly, for the outer wires, we have the following multiplicative
contributions:
\begin{equation}\alpha = (1)(2): \, d_2M^2 \qquad \text{and} \qquad \alpha =
(12): \, d_2^2M.\end{equation}
Putting together the four contributions, weighted by the corresponding
Weingarten functions, we obtained the announced formula. A derivation using the
\texttt{RTNI} software package \cite{fkn19} is provided in the Supplementary
Material.
\end{proof}

\section{Algebraic lemma and a proof of Proposition \ref{prop:Fano3}}
\label{app:lamma}

Before presenting an alternative proof of the Proposition, which relays on the
Kraus representation of the map, we formulate a useful algebraic fact concerning
trace, tensor product and reshuffling.

\begin{lemma} \label{lemma:ACB^TD^T}
	For any four square matrices $A,B,C,D$ of the same size the following
	relation holds
	\begin{equation}
	{\rm Tr} \left(  (A \otimes B) (C \otimes D)^R  \right)=
	{\rm Tr} \left( ACB^TD^T \right).
	\label{4matrices}
	\end{equation}
\end{lemma}

The above identity can be directly verified by playing with indices,
$A_{a\mu}B_{b\nu}(C_{\mu a} D_{\nu b})^R=
A_{a\mu} C_{\mu \nu} B^T_{\nu b} D^T_{b a}$,
where sums over repeating indices takes place,
but it is instructive to contemplate the proof in a form of
the following diagram depicted in Figure~\ref{fig:abcdr} .
\medskip

\begin{figure}[h!]
	\centering
	\includegraphics[width=0.7\linewidth]{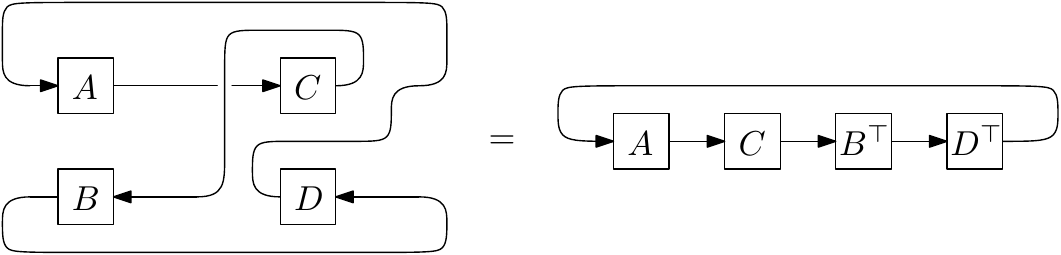}
	\caption{
		Equivalence of diagrams corresponding to ${\rm Tr}  \left( (A \otimes B)
		(C \otimes D)^R \right) $ and ${\rm Tr} \left( ACB^TD^T \right)$.
	}
	\label{fig:abcdr}
\end{figure}

\medskip

To demonstrate  Proposition \ref{prop:Fano3} we evaluate matrix elements
(\ref{super2}) of the Bloch representation  $\tilde\Phi$ of any map $\Phi$
expressed by the Kraus operators $A_\alpha$,
\begin{equation}
\begin{split}
\tilde{\Phi}_{ij} &=
\frac1d \tr (\Lambda_i \Phi(\Lambda_j)) =
\frac1d \sum_{\alpha} \tr ( \Lambda_i A_\alpha \Lambda_j A_\alpha^\dagger)
=\frac1d \sum_{\alpha} \tr \left( (\Lambda_i \otimes \Lambda_j^T) (A_\alpha
\otimes
\bar{A}_\alpha)^R \right) \\
&=\frac1d \tr \left((\Lambda_i \otimes \Lambda_j^T) \left (\sum_{\alpha}
A_\alpha
\otimes \bar{A}_\alpha \right )^R \right)
=\frac1d \tr \left( (\Lambda_i \otimes \Lambda_j^T) J_{\Phi}\right) \\
&= \frac1d \tr \left(J_{\Phi}^{T_2} (\Lambda_i \otimes \Lambda_j) \right)
=\tilde{R} \left( J_{\Phi}^{T_2}/d \right)_{ij}.
\end{split}
\end{equation}
In the first line we applied Lemma \ref{lemma:ACB^TD^T}, while in the second
line expansion (\ref{Phisum}) of the superoperator $\Phi$ in terms of the Kraus
operators $A_{\alpha}$ allowed us to arrive at the required form.

\section{Random quantum channels distributed according to other measures}
\label{app:other-measures}

 In Section \ref{sec:different-measures} we described several families of
 probability distributions for quantum channels, all of them having the flat,
 Hilbert--Schmidt (HS) measure as a special case. Various techniques to generate
 random maps according to these measures were presented and compared. In this
 Appendix, a short review of other ensembles of quantum operations will be
 presented. Some ensembles are defined on a certain subset of the entire set of
 stochastic maps, for instance the set of bistochastic maps. However, we shall
 start with an ensemble of random channels, determined by a general procedure
 leading  to a non-uniform probability measure in the space of quantum
 operations.

\bigskip

\textbf{e) Measures induced by ensembles of random density matrices}. Any
       ensemble of random states \cite{SZ04,ZPNC11}, used to generate random
       Choi matrices, defines by the rescaling \eqref{eq:normalize-Choi} an
       ensemble of random operations. For instance,  generating Wishart matrix
       according to the Bures measure, $W=XX^{\dagger}$ with $X=({\mathbbm
       I}+U)G$, where $U\in \mathcal U(d^2)$ denotes a Haar random unitary
       matrix, while  $G$ stands for a square complex Ginibre matrix of size
       $d^2$ one obtains the Bures-like measure in the space of quantum maps. A
       larger family of cognate ensembles can be obtained by applying the
       generalized Bures distributions \cite{MNPZ15}, while taking a product of
       $s$ independent random Ginibre matrices, $X=G_1G_2 \cdots G_s$ leads to
       the Fuss-Catalan distribution \cite{ZPNC11} of order $s$.
 \medskip

\textbf{  f) Random coupling with environment in a generic mixed state.} A new
family of distributions, closely related to the measures $\mu^{Stinesping}_{d_1,
d_2; M}$ from Definition \ref{def:random-Stinespring} can be obtained by
generalizing the channel from Eq.~\eqref{eq:mu-stinespring-unitary} to use a
mixed state for the environment. We obtain
\begin{equation}
\Phi_m(\rho) = [\operatorname{id}_{d_2} \otimes \Tr_{d_1}] [U (\rho\otimes
\sigma)
U^{\dagger}].
\label{eq:chan-fc}
\end{equation} Above,
we need to specify the distribution of the environment state $\sigma$, which
could be deterministic or, e.g., $\sigma=GG^{\dagger}/{\rm Tr} GG^{\dagger}$
generated according the HS measure with use of a Ginibre matrix $G$ of size
$d_2^2$. The interaction unitary $U$ is a Haar random unitary of order $d_1d_2$.

Then the Choi matrix  $J_{\Phi_m}=\Phi^R_m$ has full rank, but the probability
measure in the space of quantum operations induced in this way is not flat. For
$d_1=d_2$ and $\sigma$ sampled uniformly the distribution of eigenvalues is
given by the free product of two Mar\v{c}henko--Pastur distributions, equivalent
to the Fuss-Catalan distribution of order $k=2$.

The proof of the above fact is based on a conjecture \cite{MLZ18} that a
reshuffled Haar random unitary $U^R$ behaves asymptotically like a generic
Ginibre matrix, proved in \cite{MPS20}. For a channel $\Phi \in \CC_{d,d}$ of
the form
\begin{equation}
\Phi(\rho) = [\operatorname{id}_{d} \otimes \Tr_{d}] [U (\rho\otimes
\sigma)U^{\dagger}],
\end{equation}
where $\sigma \in M_d(\C)$ is density matrix from the HS distributions and $U$
is Haar unitary matrix of order $d^2$. We write
\begin{equation}
\begin{split}
J_\Phi&=
[\Phi \otimes \operatorname{id}_{d}](|\1_d \rangle\rangle
\langle \langle \1_d |)=
[\operatorname{id}_{d} \otimes \Tr_{d} \otimes
\operatorname{id}_{d}] [(U \otimes \1_d) (|\1_d
\rangle\rangle \langle
\langle \1_d |_{1,3} \otimes
\sigma_2)(U \otimes \1_d)^{\dagger}]\\
&=(\1_{d^2} \otimes \langle\langle
\1_d|_{1,3})(\1_d \otimes U \otimes \1_d)(\1_{d^2} \otimes
|\1_d\rangle\rangle_{2,4})(\1_d\otimes\sigma)(\1_{d^2} \otimes
\langle\langle
\1_d|_{2,4})(\1_d \otimes U^\dagger \otimes \1_d)(\1_{d^2} \otimes
|\1_d\rangle\rangle_{1,3})\\
&=U^R(\1_d \otimes \sigma)\left(U^R\right)^\dagger.
\end{split}
\end{equation}
Hence due to independence of $U$ and $\sigma$, the Choi matrix $J_\Phi$
has Fuss-Catalan distribution of order $k=2$.

\medskip

\textbf{    g) Random quantum bistochastic channels.} Consider an arbitrary
      Wishart matrix $W=XX^{\dagger}$ of order $d^2$ acting on a composite space
      $\C^{d}_A \otimes \C^d_B$, which determines by reshuffling transformation
      a completely positive map, $\Phi_0=W^R$. Rescaling it according to
      \eqref{eq:normalize-Choi} we obtain a legitimate Choi matrix $J_1$, which
      satisfies the partial trace condition, $\Tr_{A} J_1 = \mathbbm{1}_{d}$,
      and represents a trace preserving map,  $\Phi_1=J_1^{\mathrm R}$. To
      obtain a unital channel  one applies the complementary transformation
      performs on the output matrix,
\begin{equation}
\label{eq:normalize-2}
J_2=
( E^{-1/2} \otimes {\mathbbm 1}_{d} ) J_1 ( E^{-1/2} \otimes {\mathbbm 1}_{d})  ,
\end{equation}
where  $E={\Tr_{B}} J_1$ denotes the dual partial trace which forms a
semipositive matrix of size $d$. By construction $J_2$ represents a dynamical
matrix which satisfies the unitality condition of partial trace, $\Tr_{B} J_2 =
\mathbbm{1}_{d}$. Alternative iterations by  transformations
\eqref{eq:normalize-Choi} and  \eqref{eq:normalize-2} performed on an initial
random Wishart matrix $W$ converge with probability one as this scheme belongs
to the class of Sinkhorn algorithms \cite{Si64, Dj70,Gu04,BGOWW}. This procedure
can be considered as alternating projections on manifolds  \cite{LM08}, or as a
special case of the technique  \cite{DLP16} to produce a bipartite quantum state
with prescribed both partial traces. The limiting matrix $J_{\infty}$ satisfies
thus both partial trace conditions and represents a quantum map which is both
trace preserving and unital. Observe that these both transformations
correspond to a pre-- and post--processing of the initial map $\Phi_0$ by
conjugation with positive hermitian  matrices, $\Phi_2=\Psi_{E^{-1/2}} \cdot
\Phi_0 \cdot \Psi_{\bar{H}^{-1/2}}$, where   $\Psi_H(\rho)=H\rho H$ and $H > 0$.
Convergence of the iteration procedure is related to the possibility to
represent any generic completely positive map $\Phi$ as a concatenation of a
bistochastic operation $\Phi_B$ sandwiched between pre- and post-processing,
$\Phi=\Psi_E \cdot \Phi_B \cdot \Psi_H$ -- see \cite{AS15}. Fast convergence of
the above algorithm provides an efficient method to generate bistochastic random
quantum channels \cite{AS08}, analogous to the Sinkhorn method of generating
random bistochastic matrices \cite{CSBZ09}. Taking into account these results
one can expect that the measure induced in this way by the HS measure in the
space of stochastic maps will not lead to the HS  measure in the subset of
bistochastic maps.

      \medskip
\textbf{      h) Randomized unitary channels}, $\rho'=\sum_{i=1}^M p_i U_i \rho
        U_i^{\dagger}$, with independent  Haar random unitary matrices $U_i$ and
        various choices of the measure for the random probability vector $p$ of
        length $M$. These operations are bistochastic by construction. For instance,  to construct quantum expanders Hastings used such an ensemble \cite{has07} with random unitary matrices $U_i$ mixed with fixed probabilities, $p_i = 1/M$. Note that choosing $U_i$
        from a given set of  $M$ fixed (non-random) unitaries, $U_i \in \mathcal
        U(d)$, this model is equivalent to  {\sl random external fields}
        \cite{AL00} also called {\sl mixed unitary channels}, often applied in
        the theory of quantum information. For $d=2$ every bistochastic channel
        is unitarily equivalent to a Pauli map -- a mixed unitary channel, but
        for larger systems, $d\ge 3$, there exist bistochastic quantum channels
        which do not belong to this class \cite{LS93}. However, mixed unitary
        channels have non-zero volume in the set of bistochastic channels
        \cite{wat}. In more advanced approaches the independence restriction of 
        matrices $U_i$ is loosen. Dynamics of open quantum systems defined by 
        correlated, time-independent Hamiltonians were investigated in 
        \cite{HBCopen, HBCdephasing}.  

   \medskip
 \textbf{i) Random unistochastic maps.}
      These channels can be considered as a particular case of model 
      \textbf{f)},
        if the state of environment of size $M$, is the maximally mixed state, 
        $\Phi_U(\rho) = [\operatorname{id}_{d} \otimes
        \Tr_{M}] [U (\rho\otimes {\mathbbm 1}_M) U^{\dagger}]$. Such maps are
        determined only by a unitary matrix $U$ which is assumed to be random
        according to the Haar measure. Unistochastic maps are clearly unital and
        thus bistochastic for any dimension of the environment. It is convenient
        to assume that $M=kd$ and call such maps $k$--unistochastic
        \cite{ZB04}. In the simplest case, $k=1$ so that the size of the
        environment and the system are equal, the set of one-qubit unistochastic
        maps, forms a (non-convex) proper subset \cite{NA07,MKZ13} of the 
        regular
        tetrahedron of one-qubit Pauli channels. For a unistochastic channel
        $\Phi_U$  determined by a unitary matrix $U$ of size $d^2$ the
        corresponding Choi dynamical matrix is given by the reshuffled matrix,
        $J_{\Phi_U}=\frac{1}{d} U^R(U^R)^{\dagger}$ -- see \cite{BZ17}. This
        expression defines an ensemble of random unistochastic channels, as $U$
        is taken as a random unitary matrix.

   \medskip
 \textbf{j) Random POVMs.}
	In \cite{hjn}, the authors introduce the notion of \emph{random positive
	operator valued measures} (POVMs), corresponding to an ensemble of
	generalized quantum measurements. One way to construct random POVMs is to
	define POVM elements $A_1, \ldots, A_k \in M_d(\mathbb C)$ by sampling
	independent, identically distributed random Wishart matrices $W_1, \ldots,
	W_k$ and normalizing them to have unit sum \cite[Section V.B]{hjn}. This
	construction is related to that the ensembles of quantum channels considered
	in this work, in the following way. If $\Phi$ is sampled according to the
	Choi ensemble $\mu^{Choi}_{d,k;n}$ from Section \ref{sec:different-measures}
	a), then the channel $\Psi := \operatorname{diag} \circ \Phi$, where
	$\operatorname{diag}$ is the dephasing operator
	$\operatorname{diag}\ketbra{i}{j} = \delta_{ij} \ketbra{i}{j}$, has the
	distribution of a random POVM of parameters $(d,k;n)$, see \cite{hjn}.

\bigskip

Let us end this discussion by mentioning that one can consider mixtures of the
different probability measures discussed in this Appendix and of those analyzed
in the main text. In the recent work \cite{Sa20}, the authors consider a mixture
between a random unitary conjugation $X \mapsto UXU^*$ and an independent random
quantum channel as in Definition \ref{def:random-Kraus}. They study the spectral
properties of the corresponding superoperator as a function of the mixing
parameter.

%%%%%%%%%%%%%%%%%%%%%%%%%%%%%%%%%%%%%%%%%

\end{document}